\documentclass[conference,10pt]{IEEEtran}

\usepackage{amsmath}
\usepackage{stmaryrd}

\usepackage{array}

\usepackage[hyphens]{url}
\usepackage[hidelinks]{hyperref}

\usepackage{amsthm,amssymb}
\newtheorem{theorem}{Theorem}
\newtheorem{proposition}[theorem]{Proposition}
\newtheorem{fact}[theorem]{Fact}
\newtheorem{lemma}[theorem]{Lemma}
\newtheorem{corollary}[theorem]{Corollary}
\theoremstyle{definition}
\newtheorem{definition}[theorem]{Definition}
\newtheorem{notation}[theorem]{Notation}
\newtheorem{remark}[theorem]{Remark}

\usepackage{commath}
\usepackage{cmll}

\usepackage{ebproof}

\newcommand{\placeholder}{-}

\newcommand{\R}{\mathbb{R}}
\newcommand{\N}{\mathbb{N}}
\newcommand{\W}{\mathbb{W}}
\newcommand{\E}{\mathbb{E}}

\DeclareMathAlphabet{\mymathbb}{U}{BOONDOX-ds}{m}{n}
\newcommand{\zero}{\mymathbb{0}}

\newcommand{\emb}[1]{\underline{#1}}
\newcommand{\baseQbs}[1]{\left\vert#1\right\vert}

\newcommand{\QbsConv}{\operatorname{QbsConv}}
\newcommand{\linear}{\multimap}
\newcommand{\analytic}{\Rightarrow}

\newcommand{\J}{\mathcal{J}}
\newcommand{\K}{\mathcal{K}}
\newcommand{\M}{\mathcal{M}}
\newcommand{\NN}{\mathcal{N}}
\newcommand{\D}{\mathcal{D}}

\newcommand{\inj}[2]{\operatorname{Inj}(#1,#2)}

\newcommand{\denot}[1]{\left\llbracket#1\right\rrbracket}

\newcommand{\kw}[1]{\mathrel{\operatorname{#1}}}

\begin{document}
%
\title{Extensional Denotational Semantics of Higher-Order Probabilistic Programs, Beyond the Discrete Case}


\author{\IEEEauthorblockN{Guillaume Geoffroy}
\IEEEauthorblockA{Università di Bologna}}


\maketitle

\begin{abstract}
We describe a mathematical structure that can give extensional denotational semantics to higher-order probabilistic programs. It is not limited to discrete probabilities, and it is compatible with integration in a way the models that have been proposed before are not. It is organised as a model of propositional linear logic in which all the connectives have intuitive probabilistic interpretations. In addition, it has least fixed points for all maps, so it can interpret recursion.
\end{abstract}


%
\IEEEpeerreviewmaketitle

\section{Introduction}

Extensional denotational semantics of probabilistic programs have been around since the late 1970s \cite{djahromi80, kozen81}. The idea is to represent first-order programs by measure transformers. For example, consider a program that takes as input a handle to a random generator of real numbers and outputs a randomly chosen real number: its denotation is a map that takes a probability measure on $\R$ (representing the distribution of the input) and returns a probability measure on $\R$ (representing the distribution of the output). 
If non-terminating programs are allowed, then instead, you get a map from sub-probability measures to sub-probability measures: a total measure strictly less than $1$ denotes a non-zero probability of failing to produce a number (\textit{e.g.} through non-termination). Beyond first order, ``extensional'' means that each type is interpreted as a set (in the case of the real numbers, the set $S(\R)$ of sub-probability distributions on $\R$) with some additional structure, and programs are interpreted as structure-preserving maps, in such a way that the set of all structure-preserving maps can itself be equipped with the same structure. What structure should one use, though? This is a long-standing question: interaction between the ``set of all sub-probability mesures'' construction and the domain-theoretic tools traditionally used to represent function types has been notoriously troublesome \cite{troublesome98}. As a result, for a while, more success was to be had from ``intensional'' approaches, such as games \cite{probaGames02} and the geometry of interaction \cite{dlh-bayesian-19}.

Some hints at how to answer this question can be found in Kozen's early extensional model \cite{kozen81}. This model can only represent first-order programs, and only those that are ``linear'', in the sense that they sample their input exactly once. Measure transformers that are \emph{linear} and \emph{effective}, in the sense that they actually represent one of these linear programs, have a particular property: they \emph{commute with integrals}. Equivalently, they correspond to sub-probability kernels \cite{kallenberg17} (through the ``bind'' operation of the monad $S$ of sub-probability measures \cite{categoryMarkov99}). Still equivalently, they are morphisms of algebras over this monad. This implies that, as noted by Kozen back then (though with different terminology), two linear effective functions that coincide on all Dirac distributions (\textit{i.e.} on all deterministic inputs) are equal. This leaves out many maps; in particular, the map that takes a sub-probability measure on $\R$ and returns its continuous part is not linear effective.

If one extends Kozen's model to allow programs to sample their input any number of times, effective functions start to look like power series. Indeed, in this modified model, it can be checked that for all effective maps $f:S(\R)\to S(\R)$, there exists a family $(\partial_n f : S(\R^n) \to S(\R))_{n\in\N}$ of \emph{linear} effective maps (so morphisms of $S$-algebras) such that for all $\mu\in S(\R)$, $f(\mu) = \sum_{n\in \N} \partial_n f(\mu^{\otimes n})$ (where $\mu^{\otimes n}$ denotes the product measure of $n$ copies of $\mu$). As a result, it is easy to prove that two effective maps that coincide on all finitely-supported measures are equal. A fortiori, the ``continuous part'' map defined above, which coincides with the zero map on all discrete measures, is not effective.

This suggests that, in order to get well-behaved extensional denotational semantics, types in general should be interpreted by $S$-algebras, linear programs by morphisms of $S$-algebras, and general programs by ``power series'' in the sense described above. Probabilistic coherence spaces \cite{pcoh11} do precisely that. Among their many pleasant properties, they form a model of full propositional linear logic \cite{girardLL87} in which each connective (or at least a complete set thereof) has an intuitive probabilistic interpretation. A considerable drawback is that they are designed to only work with countable data types, and therefore discrete probabilities. Indeed, the only reason why linear maps between probabilistic coherence spaces commute with integrals (\textit{i.e.} are morphisms of $S$-algebras) is because in this context, integrals are just countable sums.

Several constructions have been proposed to overcome this limitation \cite{cones17, omegaQbs18, slavnov18, roban19}. However, they do so only partially. Indeed, none of them fits the above picture of $S$-algebras and power series, as can be seen by the fact that they all include the ``continuous part'' map mentioned above. This means that none of these models is compatible with integration (only with countable sums, at most), even though integration is a cornerstone of probability theory.

The contribution of this paper is to define such a model, which we call \emph{convex quasi-Borel spaces}. This brings us one step closer to answering the long-standing question from the first paragraph. In addition to compatibility with integration, the construction we propose is a model of linear logic (with the same intuitions behind the connectives as in probabilistic coherence spaces), and all functions have least fixed points (so the model interprets recursion). Our construction relies on two innovations with respect to previous models: we define integration axiomatically and simply ask that linear maps commute with integrals; and we do away with topology entirely -- in particular we replace all limits and suprema with countable sums (which are just a particular kind of integrals). The idea behind the second point is that pointwise countable sums interact well with power series, whereas pointwise limits do not.

We begin this paper with a few reminders on \emph{quasi-Borel spaces} \cite{qbs17} (hereafter QBSs), which are a generalisation of the traditional notion of measurable spaces. They form a category that contains the category of measurable spaces and measurable maps as a subcategory (and the category of standard Borel spaces as a full subcategory), supports a commutative ``set of all probability measures'' strong monad, and is cartesian closed. This means that, by themselves, QBSs can already be seen as a model of the simply typed lambda calculus that supports probabilistic constructions (however, this model pays no attention whatsoever to computability: it contains not just the ``continuous part'' map, but in fact any measurable map from $S(\mathbb R)$ to $S(\mathbb R)$). They can also be seen as an alternative theory of integration and measurability that supports function spaces, which is why our construction will be built upon them.

The remainder of the paper is mostly dedicated to the construction of the model itself. This construction follows the blueprint of linear logic \cite{girardLL87}. First, we define convex QBSs, and as a canonical example, we define for all QBSs $A$ the convex QBS $\emb A$ of random elements of $A$. Then we define the multiplicative constructions, notably \emph{multilinear maps}, which represent probabilistic programs with multiple arguments that use each argument exactly once. In particular, we show that convex QBSs and linear maps form a closed symmetric monoidal category. Next come the additive constructions (the cartesian product and the coproduct), and then the exponential modalities. After that, we define \emph{analytic maps}, which represent probabilistic programs in general, and we show that the category of convex QBSs and analytic maps is equivalent to the Kleisli category of the exponential comonad ``$!$'' -- in particular, it is cartesian closed. Finally, we show that all analytic maps from a convex QBS to itself have a least fixed point, and that the operation that maps a map to its least fixed point is analytic. After building the model, we give an example of toy probabilistic language that can be interpreted in it.

Note that we use the expression ``probabilistic program'' in its narrow meaning of a program that can make random choices (as well as manipulate other programs that make their own random choices). The same expression is also used in the broader sense of a program that describes and manipulates statistical models \cite{statonESOP17, conditional19}. In that setting, compatibility with integration is all the more relevant, so it would be worth checking whether our approach can be generalised to it.

\section{Preliminaries on quasi-Borel spaces}

Getting a construction that is compatible with integration requires a theory of integration and measurability in the first place. Instead of the traditional theory of $\sigma$-algebras and measurable spaces, we use quasi-Borel spaces \cite{qbs17} (QBSs), because they are known to form a cartesian closed category.

The only cost of replacing measurable spaces with QBSs is that the ``source of randomness'' has to be a standard Borel space \cite{kallenberg17}. This is a reasonable assumption as far as computer science is concerned, since that includes the space $\{0,1\}^\N$ of all infinite sequences of bits\footnote{In fact, by standard results on Markov kernels, the probability measures that can be represented in QBSs are exactly those that can be obtained by pushing forward the usual ``independent fair coins'' measure on $\{0,1\}^\N$.}.

In this section, we recall the definition of QBSs and define the monad $S$ of \emph{sub-probability measures}.

\begin{definition}[\cite{qbs17}] A \emph{quasi-Borel space} is the data of a set $A$ and a set $M_A$ of maps from $\R$ to $A$ such that
\begin{itemize}
  \item for all $\alpha \in M_A$ and all measurable maps $f : \R \to \R$ (where $\R$ is equipped with the Borel $\sigma$-algebra $\Sigma_\R$), $\alpha \circ f \in M_A$,
	\item for all constant maps $\alpha : \R \to A$, $\alpha \in M_A$,
	\item for all $\left(\alpha_n\right)_{n \in \N} \in M_A^\N$ and all partitions $\left(U_n\right)_{n \in \N}$ of $\R$ into Borel sets, $M_A$ contains the case-split map that maps $r$ to $\alpha_n(r)$ whenever $r \in U_n$.
\end{itemize}

Let $A$ and $B$ be QBSs. A \emph{morphism of QBSs} from $A$ to $B$ is a map $f : A \to B$ such that for all $\alpha \in M_A$, $f \circ \alpha \in M_B$.
QBSs and morphisms between them form a category, which we denote by $\operatorname{Qbs}$.
\end{definition}

Since this category is cartesian \cite[Proposition 16]{qbs17} (and even cartesian closed), it induces a symmetric multicategory \cite[Definitions 2.1.1 and 2.2.21]{leinster03}, which we also denote by $\operatorname{Qbs}$, with the set of $n$-ary maps $\operatorname{Qbs}(A_1, \ldots, A_n; B)$ defined as $\operatorname{Qbs}(A_1 \times \ldots \times A_n, B)$.

For all QBSs $A$ and all sets $B \subseteq A$, we define the \emph{subset QBS structure} on $B$ by $M_B = \{ \alpha \in M_A;~ \forall r\in \R, \alpha(r) \in B \}$.

Measurable spaces and measurable maps form a subcategory of $\operatorname{Qbs}$: for all measurable spaces $\left(A, \Sigma_A\right)$, we define a QBS $\left(A, M_A\right)$ by letting $M_A$ be the set of all measurable maps from $\left(\R, \Sigma_\R\right)$ to $\left(A, \Sigma_A\right)$. If $\left(A, \Sigma_A\right)$ and $\left(B, \Sigma_B\right)$ are measurable spaces, every measurable map from $\left(A, \Sigma_A\right)$ to $\left(B, \Sigma_B\right)$ is a morphism of QBSs from $\left(A, M_A\right)$ to $\left(B, M_B\right)$, though the converse is not necessarily true. 

Recall \cite[Introduction]{kallenberg17} that a \emph{standard Borel space} (or simply \emph{Borel space}) is a measurable space that is isomorphic to a Borel subset of $\R$. If $\left(A, \Sigma_A\right)$ and $\left(B, \Sigma_B\right)$ are standard Borel spaces, then a map $f : A \to B$ is measurable if and only if it is a morphism of QBSs \cite[Proposition 15]{qbs17}. From now on, we will only consider QBSs and standard Borel spaces (seen as a particular case of QBSs), and never deal with general measurable spaces. As a result, we will refer to morphisms of QBSs simply as \emph{measurable maps}: this will help convey the right intuitions, and it will never come in conflict with the usual notion of measurability.

For all standard Borel spaces $\left(A, \Sigma_A\right)$, we denote by $G(A)$ the set of all probability measures on $A$. We equip it with the smallest $\sigma$-algebra $\Sigma_{G\left(A\right)}$ such that for all $U \in \Sigma_A$, $\mu \mapsto  \mu{\left(U\right)}$ is measurable. Recall \cite{giry82} that $G(A)$ is itself a standard Borel space, and that a measurable map from $A$ to $G(B)$ is the same thing as a Markov kernel from $A$ to $B$.
This construction has been successfully generalised to QBSs \cite[Section V-D]{qbs17}. We describe an analogous construction for sub-probability measures, \textit{i.e.} positive measures of total weight at most $1$. With the exceptions of Facts \ref{fact:S-M} and \ref{fact:S-sum}, all the results we give here correspond to results that have already been established in the case of probability measures, so we omit their proofs, which are similar.

\begin{definition}\label{def:S} Let $A$ be a QBS. The QBS of \emph{sub-probability measures} on $A$, denoted by $S(A)$, is defined as the quotient \cite[Proposition 25]{qbs17} \[S(A) = \operatorname{Qbs}\left(\R, A \amalg \left\{*\right\} \right) \times G(\R) / \sim,\] where $\operatorname{Qbs}(B_1,B_2)$ denotes the QBS of measurable maps from $B_1$ to $B_2$ \cite[Proposition 18]{qbs17}, $\amalg$ denotes the coproduct of QBSs \cite[Proposition 17]{qbs17}, $\{*\}$ denotes the one-element QBS, and $\sim$ denotes the following equivalence relation:

For all $(\alpha, \mu), (\beta, \nu)$, we let $(\alpha, \mu) \sim (\beta, \nu)$ if and only if for all measurable maps $f : A \to \left[0, +\infty\right]$,
\[\int_{\scriptscriptstyle r \in \alpha^{-1}\left(A\right)} \hspace{-1cm} f (\alpha(r))\, \mu(\dif r) = \int_{\scriptscriptstyle s \in \beta^{-1}\left(A\right)} \hspace{-1cm} f (\beta(s))\, \nu(\dif s).\]

We denote by $\left[\alpha, \mu\right]$ the equivalence class of $(\alpha, \mu)$.
\end{definition}

This approach of pushing forward by a partial map (or equivalently, a map that can take the ``undefined'' value $*$) in order to ``shave'' some of the original measure has already been used to define an analogue of the monad $S$ in the context of $\omega$-QBSs \cite{omegaQbs18}.

\begin{definition} We make $S:\operatorname{Qbs} \to \operatorname{Qbs}$ into a functor by letting $S(f)\left(\left[\alpha, \mu\right]\right) = \left[f \circ \alpha, \mu\right]$ for all $f : A \to B$ measurable and all $\left[\alpha, \mu\right] \in S(A)$.
\end{definition}

For all QBSs $A$ and $B$, it is clear that $S$ defines a measurable map from $\operatorname{Qbs}\left(A, B\right)$ to $\operatorname{Qbs}\left(S(A), S(B)\right)$.

\begin{definition}[Integration on QBSs] \label{qbs-integration} Let $A$ be a QBS. For all $\rho = \left[\alpha, \mu\right] \in S(A)$ and all $f \in \operatorname{Qbs}\left(A, [0, +\infty]\right)$, we let
\[\int_{\scriptscriptstyle x \in A}  \hspace{-.35cm} f(x)\, \rho(\dif x) = \int_{\scriptscriptstyle r \in \alpha^{-1}\left(A\right)} \hspace{-1cm} f (\alpha (r))\, \mu(\dif r) \in [0, +\infty].\]
\end{definition}

One can check that this defines a measurable map from $S(A) \times \operatorname{Qbs}\left(A, [0, +\infty]\right)$ to $[0, +\infty]$.

\begin{fact} For all standard Borel spaces $\left(A, \Sigma_A\right)$, the following map defines a bijection between the QBS $S{\left(A, M_A\right)}$ and the set of all sub-probability measures (in the traditional sense) on the standard Borel space $\left(A, \Sigma_A\right)$:
\[\left\{ \begin{array}{ccccc} S(A, M_A) & \to & \Sigma_A & \to & [0,1] \\ \rho & \mapsto & U & \mapsto & \int_{x \in A} \mathbf{1}_U(x)\, \rho(\dif x),\end{array} \right.\]
where $\mathbf{1}_U$ denotes the indicator function of $U$. In addition, this bijection is natural in $A$, and integration as in Definition \ref{qbs-integration} corresponds to integration in the traditional sense through this bijection.
\end{fact}

If we unfold Definition \ref{def:S}, we find that $M_{S\left(A\right)} = \{ r \mapsto \left[s \mapsto \alpha(r,s), \mu_r\right];~ \alpha \in \operatorname{Qbs}\left(\R \times \R, A \amalg \left\{*\right\} \right), \left(r \mapsto \mu_r\right) \in \operatorname{Qbs}\left(\R, G\left(\R\right)\right) \}$. However, in order to define the monad multiplication, we will need the following characterisation\footnote{In fact, the original paper on QBSs goes the other way round. It defines the monad $P$ of probability measures similarly to Fact \ref{fact:S-M}, and then it lays out the arguments needed to prove a characterisation in the spirit of Definition \ref{def:S} \cite[proof of Lemma 27]{qbs17}.}.

\begin{fact} \label{fact:S-M}
For all QBSs $A$, $M_{S(A)} = \{ r \mapsto \left[\alpha, \mu_r\right];~ \alpha \in \operatorname{Qbs}\left(\R, A \amalg \left\{*\right\} \right), \left(r \mapsto \mu_r\right) \in \operatorname{Qbs}\left(\R, G(\R)\right) \}$.
\end{fact}
\begin{proof} Let $\alpha \in \operatorname{Qbs}\left(\R \times \R, A \amalg \left\{*\right\} \right)$ and $\left(r \mapsto \mu_r\right) \in \operatorname{Qbs}\left(\R, G(\R)\right)$. The maps $r \mapsto \mu_r$ and $r \mapsto \delta_\R(r)$ are Markov kernels (where $\delta_\R(r)$ denotes the Dirac measure at $r$ on $\R$), so $r \mapsto \delta_\R(r) \otimes \mu_r$ is a Markov kernel \cite[Lemma 1.17, as a particular case of composition where the second kernel ignores its second argument]{kallenberg17} and therefore a measurable map from $\R$ to $G(\R \times \R)$. Let $\varphi$ be an isomorphism between the standard Borel spaces $\R \times \R$ and $\R$. Then for all $r \in \R$, $\left[s \mapsto \alpha(r,s), \mu_r\right] = \left[\alpha \circ \varphi^{-1}, \varphi_\sharp {\left(\delta_\R(r) \otimes \mu_r\right)}\right]$, where $\varphi_\sharp \left(\delta_\R(r) \otimes \mu_r\right)$ denotes the pushforward measure of $\delta_\R(r) \otimes \mu_r$ by $\varphi$.
\end{proof}

In particular, for all $\rho \in S(S(A))$, there exists $\mu \in G(\R)$, $U \in \Sigma_\R$, $(r\mapsto\nu_r)\in\operatorname{Qbs}\left(U, G\left(\R\right)\right)$ and $\alpha \in \operatorname{Qbs}\left(\R, A \amalg \left\{*\right\} \right)$ such that
\[\rho = \left[ r \mapsto \left(\begin{array}{cl} \left[ \alpha, \nu_r \right] & \text{if } r \in U \\  * & \text{otherwise} \end{array}\right), \mu \right].\]

\begin{definition}\label{def:S-monad} We make $S$ into a monad $\left(S, \delta, \E\right)$ as follows\footnote{$\delta$ stands for \emph{Dirac}, and $\E$ stands for \emph{expected value}.}.
For all QBSs $A$ and all $x\in A$, $\delta_{A}(x) = [r \mapsto x, \mu]\in S(A),$
 where $\mu$ is any probability measure on $\R$. For all \[\rho = \left[ r \mapsto \left(\begin{array}{cl} \left[ \alpha, \nu_r \right] & \text{if } r \in U \\  * & \text{otherwise} \end{array}\right), \mu \right]\in S(S(A)),\]
\[\E_{A}(\rho) = \left[\alpha, V \mapsto \int_{\scriptscriptstyle r\in U} \hspace{-.3cm} \nu_r(V)\,\mu(\dif r)\right]\in S(A).\]
Alternatively, we write $\emb{x}$ for $\delta_{A}(x)$ and $\int_{\tau \in G(A)} \tau\,\rho(\dif \tau)$ for $\E_{A}(\rho)$.
\end{definition}

This means that sub-probability distributions are stable under sub-convex combinations in a very broad sense, and that this
``sub-convex combination'' (or ``expected value'') operation is measurable. In addition, sub-probability distributions are stable under countable sums as long as the sum of the total weights is at most $1$, and this ``countable sum'' operation is also measurable on its domain:

\begin{fact}\label{fact:S-sum} Let $A$ be a QBS. For all $(\rho_{n})_{n\in\N} \in S(A)^\N$ such that $\sum_{n \in \mathbb N} \int_{x \in A} 1\,\rho_{n}(\dif x) \leq 1$, there exists a unique $\left(\sum_{n\in\N} \rho_n\right) \in S(A)$ such that for all measurable maps $f : A \to [0,+\infty]$, \[\int_{x\in A} f(x)\,{\left(\sum_{n\in\N} \rho_n\right)}(\dif x) = \sum_{n\in\N}  \int_{x\in A} f(x)\,\rho_n(\dif x).\] In addition, the map $(\rho_{n})_{n\in\N} \mapsto \sum_{n\in\N} \rho_n$ is measurable (when its domain is equipped with the subset QBS structure).
\end{fact}
\begin{proof} 
Let $(r \mapsto \rho_{n,r})_{n\in\N} \in M_{S\left(A\right)}^\N$ such that $\sum_{n \in \mathbb N} \int_{x \in A} 1\,\rho_{n,r}(\dif x) \leq 1$ for all $r\in\R$. We must prove that there exists a unique $\left(r \mapsto \tau_r\right)\in M_{S\left(A\right)}$ such that for all $f : A \to [0,+\infty]$ measurable and all $r\in\R$, $\int_{x \in A} f(x)\, \tau_r(\dif x) = \sum_{n \in \mathbb N} \int_{x \in A} f(x)\, \rho_{n,r}(\dif x)$.

For all $n$, let $\left(r \mapsto \left[\alpha_n, \mu_{n,r}\right]\right) = \left(r \mapsto \rho_{n,r}\right)$, and let $\varphi_n$ be an isomorphism between the standard Borel spaces $\R$ and $(n,n+1]$. For all $n\in\N$ and all $r\in\R$, let $\nu_{n,r}$ be the measure on $\R$ defined by $\nu_{n,r}\left(U\right) = \mu_{n,r}\left(\varphi_n^{-1}\left(U\right) \cap \alpha_n^{-1}\left(A\right)\right)$. For all $r\in\R$, let $\nu_{*,r} = \left(1-\sum_{n\in\N} \mu_{n,r}\left(\alpha_n^{-1}\left(A\right)\right)\right)\,\delta_\R(0)$. For all $r\in\R$, we let $\nu_{\infty,r} = \nu_{*,r} + \sum_{n\in\N} \nu_{n,r}$: $\left(r\mapsto\nu_{\infty,r}\right)$ is a measurable map from $\R$ to $G(\R)$.

For all $s \in \mathbb \R$, we let $\beta(s) = \alpha\circ\varphi_n^{-1}(s)$ if $s \in (n,n+1]$ and $\beta(s) = *$ if $s \leq 0$. Then $\tau_{r} = \left[\beta, \nu_{\infty,r}\right]$ satisfies the requirements.
\end{proof}

Spaces of the form $S(A)$ will serve as a model for convex QBSs. As a result, these two properties (the existence of a measurable ``sub-convex combination'' map and of a measurable ``countable sum'' partial map) will become the main axioms of convex QBSs.

\section{Convex quasi-Borel spaces}

The only thing one can do with a probabilistic program (or any program for that matter) is to place it in some context that has an observable outcome (such as producing a real value), and observe. What happens then constitutes the \emph{behaviour} of the program. If programs can test values for equality with any constant, then it is sufficient to restrict the notion of observable outcome to just termination ($=$ success) or non-termination ($=$ failure).

The idea behind convex quasi-Borel spaces is to have a set of \emph{random values} (representing programs) and a set of \emph{linear tests} (representing contexts that use the program exactly once). For each linear test $\eta$ and each random value $x$, the structure gives a \emph{probability of success} $\eta x \in [0,1]$.
 
\begin{definition}
A \emph{convex quasi-Borel space} $X$ is the data of
\begin{itemize}
\item two QBSs $\baseQbs{X}$ (random values) and $\baseQbs{X^\bot}$ (linear tests),
\item a measurable map
\[
\cdot_{X}:\left\{ \begin{array}{ccc}
\baseQbs{X^\bot}\times\baseQbs{X} & \to & \left[0,1\right]\\
\left(\eta,x\right) & \mapsto & \eta x = \eta\cdot_{X}x
\end{array}\right.\text{,}
\]
\end{itemize}
such that
\begin{itemize}
\item for all $x,y\in\baseQbs{X}$, if $\forall\eta\in\baseQbs{X^\bot},\eta x=\eta y$,
then $x=y$,
\item there exists a (necessarily unique) measurable map $\E_X:S(\baseQbs{X})\to\baseQbs{X}$
such that for all $\mu\in S(\baseQbs{X})$ and all $\eta\in\baseQbs{X^\bot}$,
$\eta(\E_X(\mu))=\int_{x\in\baseQbs{X}}\eta x\,\mu(\dif x)$,
\item there exists a (necessarily unique) measurable map 
\[
\left\{ \begin{array}{ccc}
\left\{ \begin{array}{c}\left(x_{n}\right)_{n\in\mathbb{N}}\in\baseQbs{X}^{\mathbb{N}};\\ \forall\eta\in\baseQbs{X^\bot},\sum_{n\in\mathbb{N}}\eta x_{n}\leq1\end{array}\right\}  & \to & \baseQbs{X}\\
\left(x_{n}\right)_{n\in\mathbb{N}} & \mapsto & \sum_{n\in\mathbb{N}}x_{n}
\end{array}\right.
\]
such that for all $\left(x_{n}\right)_{n\in\mathbb{N}}$ in its domain, $\eta\left(\sum_{n\in\mathbb{N}}x_{n}\right)=\sum_{n\in\mathbb{N}}\eta x_{n}$,
\end{itemize}
and such that, symmetrically,
\begin{itemize}
\item for all $\eta,\xi\in\baseQbs{X^\bot}$, if $\forall x\in\baseQbs{X},\eta x=\xi x$,
then $\eta=\xi$,
\item there exists a (necessarily unique) measurable map $\E_{X^\bot}:S(\baseQbs{X^\bot})\to\baseQbs{X^\bot}$
such that for all $\rho\in S(\baseQbs{X^\bot})$ and
all $x\in\baseQbs{X}$, $\left(\E_{X^\bot}(\rho)\right)x=\int_{\eta\in\baseQbs{X^\bot}}\eta x\,\rho(\dif\eta)$,
\item there exists a (necessarily unique) measurable map 
\[
\left\{ \begin{array}{ccc}
\left\{ \begin{array}{c}\left(\eta_{n}\right)_{n\in\mathbb{N}}\in\baseQbs{X^\bot}^{\mathbb{N}}; \\ \forall x\in\baseQbs{X},\sum_{n\in\mathbb{N}}\eta_{n}x\leq1 \end{array}\right\}  & \to & \baseQbs{X^\bot}\\
\left(\eta_{n}\right)_{n\in\mathbb{N}} & \mapsto & \sum_{n\in\mathbb{N}}\eta_{n}
\end{array}\right.
\]
such that for all $\left(\eta_{n}\right)_{n\in\mathbb{N}}$ in its domain, $\left(\sum_{n\in\mathbb{N}}\eta_{n}\right)x=\sum_{n\in\mathbb{N}}\eta_{n}x$.
\end{itemize}
\end{definition}

In particular, if $X$ is a convex QBS, then one can easily check that $\left(\baseQbs{X}, \E_X\right)$ is an algebra over the monad $S$.

\begin{remark} It would make sense to merge the two conditions (existence and measurability of integrals of sub-probability measures on one hand, and of countable sums on the other) and ask directly for the existence and measurability of integrals of $s$-finite measures \cite{kallenberg17}. The above definition even suggests how, in this context, to represent $s$-finite measures as a QBS (namely, as a quotient of $S(A)^\N$, with $(\mu_n)_{n\in\N}$ interpreted as ``$\sum_{n\in\N} \mu_n$''). However, countable sums and sub-probability measures, taken separately, are simpler and, importantly, more widely known than $s$-finite measures: we made the choice of sacrificing some concision to gain in technical simplicity.
\end{remark}

Since the above definition is symmetric, each convex QBS comes with a \emph{dual}:

\begin{definition}
For all convex QBSs $X$, we define a convex QBS $X^{\bot}$ by letting
\begin{itemize}
\item $\baseQbs{X^{\bot\bot}}=\baseQbs{X}$,
\item for all $x\in\left|X^{\bot\bot}\right|$ and all $\eta\in\baseQbs{X^\bot}$,
$x\cdot_{X^{\bot}}\eta=\eta\cdot_{X}x$.
\end{itemize}
\end{definition}

\begin{notation}[Integration in convex QBSs] For all QBSs $A$, all convex QBSs $X$, all measurable maps $f : A \to \baseQbs{X}$ and all $\mu \in S(A)$, we write $\int_{a \in A} f(a)\,\mu(\dif a)$ for $\E(S(f)(\mu))$.
\end{notation}

Here are the main, basic examples of convex QBSs (in fact, the first two can be seen as particular instances of the third).

\begin{definition}[Multiplicative unit]
We define a convex QBS $\W$ (for \emph{weights}) by:
\begin{itemize}
\item $\left|\W\right|=\left|\W^{\bot}\right|=\left[0,1\right]$,
\item for all $\eta\in\left|\W^{\bot}\right|$ and all $x\in\left|\W\right|$,
$\eta x$ is the product of $\eta$ and $x$ as elements of $\R$.
\end{itemize}
\end{definition}

\begin{definition}[Additive unit]
We define a convex QBS $\zero$ by:
\begin{itemize}
\item $\left|\zero\right|=\left|\zero^{\bot}\right|=\left\{ 0\right\} $,
\item for all $\eta\in\left|\zero^{\bot}\right|$ and all $x\in\left|\zero\right|$,
$\eta x=0$.
\end{itemize}
\end{definition}

\begin{definition}[Data types]
For all QBSs $A$, we define a convex QBS $\emb A$ by:
\begin{itemize}
\item $\baseQbs{\emb A}=S(A)$,
\item $\baseQbs{\emb A^{\bot}}=\operatorname{Qbs}\left(A,\left[0,1\right]\right)$,
\item for all $\eta\in\baseQbs{\emb A^{\bot}}$ and all $\mu\in\baseQbs{\emb A}$,
$\eta \mu=\int_{x\in A}\eta(x)\mu(\dif x)$.
\end{itemize}
\end{definition}
Note that the existence and measurability of expected values and countable sums on $\baseQbs{\emb A}$ are given by Definition
\ref{def:S-monad} and Fact \ref{fact:S-sum}, and that expected values and countable sums on $\baseQbs{\emb A^\bot}$
are computed pointwise.

Whenever we have two expressions $\Theta_1$ and $\Theta_2$ that are not necessarily defined (such as sums of elements of $\baseQbs{X}$ for some convex QBS $X$), we will write $\Theta_1 = \Theta_2$ for ``$\Theta_1$ is defined if and only if $\Theta_2$ is, and in that case they are equal''.

\begin{notation}
For all convex QBSs $X$, all $x, y \in \baseQbs{X}$ and all $r \in [0,+\infty)$, we write \begin{itemize}
\item $0_X$ for the unique element of $\baseQbs{X}$ such that $\eta 0_X = 0$ for all $\eta \in \baseQbs{X^\bot}$,
\item $x + y$ for the unique element of $\baseQbs{X}$ such that $\eta (x+y) = \eta x+\eta y$ for all $\eta \in \baseQbs{X^\bot}$, if it exists,
\item $r x$ for the unique element of $\baseQbs{X}$ such that $\eta (r x) = r\,\eta x$ for all $\eta \in \baseQbs{X^\bot}$, if it exists,
\item $x \leq y$ if $\eta x \leq \eta y$ for all $\eta \in \baseQbs{X^\bot}$ (which defines a partial order on $\baseQbs{X}$),
\item $\norm{x}$ for $\sup_{\eta \in  \baseQbs{X^\bot}}\eta x \in [0,1]$.
\end{itemize}
\end{notation}

Note that the map $\norm{-} : \baseQbs{X} \to [0,1]$ is not measurable in general, which limits its usefulness.

\begin{fact}
The binary sum and scalar multiplication are measurable on their domains of definition (which are subsets of $\baseQbs{X} \times \baseQbs{X}$ and $[0,+\infty) \times \baseQbs{X}$ respectively).
\end{fact}
\begin{proof}
First, one can check that for all QBSs $A$, all $\mu \in S(A)$ and all $s \in[0,1]$, there exists a unique $s \mu \in S(A)$ such that for all $f \in \operatorname{Qbs}\left(A, [0, +\infty]\right)$, $\int_{a\in A} f(a)\, (s\mu)(\dif a) = s \int_{a\in A} f(a)\, \mu(\dif a)$. In addition, one can check that the operation $\left(s, \mu\right) \mapsto s\mu$ is measurable.

We have $x + y = \sum_{n\in\N} z_n$ , where $z_0 = x$, $z_1 = y$, and $z_n = 0_X$ for $n > 1$, so the binary sum is measurable on its domain.

We have $r x = \sum_{n\in\N} \E_X{\left(r_n \delta_{\baseQbs{X}}(x)\right)}$, where $r_n$ is $0$ if $r < n$, $r-n$ if $n \leq r \leq n+1$, and $1$ if $n+1 < r$. The map $r \mapsto \left(r_n\right)_{n\in\N}$ is measurable, as are $\delta_{\baseQbs{X}}$ and $\E_X$, so scalar multiplication is measurable on its domain.
\end{proof}

From the above proof, one also deduces that $r x$ is always defined when $r \leq 1$.

When $X$ is a convex QBS, we will generally write $x \in X$ for $x \in \baseQbs{X}$.

\section{Multilinear maps and multiplicative connectives}

In this section, we define \emph{multilinear maps} between convex QBSs, and we define a structure of convex QBS on spaces of multilinear maps. This construction generates all the multiplicative connectives. Intuitively, an $n$\emph{-linear map} from $X_1, \ldots, X_n$ to $Y$ represents a probabilistic program that takes $n$ arguments of types $X_1, \ldots, X_n$, uses (\textit{i.e.} samples) each one exactly once, and returns a result of type $Y$.

Given a convex QBS $X$, each linear test $\eta \in \baseQbs{X^\bot}$ can be seen as a map from $\baseQbs{X}$ to $[0,1] = \baseQbs{\W}$. Naturally, the set of ($1$-)\emph{linear maps} from $X$ to $\W$ will be defined as the set of all maps from $\baseQbs{X}$ to $\baseQbs{\W}$ that come from some $\eta \in \baseQbs{X^\bot}$. Multilinearity in general should be preserved by composition, so its definition should at least ensure that \begin{itemize}
 \item if a map $f : \baseQbs{X} \to \baseQbs{Y}$ is linear, then for all $\eta \in \baseQbs{Y^\bot}$, there exists $\xi \in \baseQbs{X^\bot}$ such that $\xi x = \eta f(x)$ for all $x$,
 \item if a map $f : \baseQbs{X_n}\times\ldots\times \baseQbs{X_n} \to \baseQbs{Y}$ is $n$-linear, then it is linear with respect to each argument.
 \end{itemize}
It would be tempting to turn these two ``if''s into ``if and only if''s and use that as the definition of multilinearity. However, we must also add conditions of measurability:

\begin{definition} Let $n$ be a natural number and $X_1, \ldots, X_n, Y$ convex QBSs. An \emph{$n$-linear map} from $X_1, \ldots, X_n$ to $Y$ is a measurable map $f : \baseQbs{X_1} \times \ldots \times \baseQbs{X_n} \to \baseQbs{Y}$ such that for all $1 \leq k \leq n$, there exists a (necessarily unique) measurable map \[f^{\bot k} : \baseQbs{X_1} \times \ldots \times \baseQbs{X_{k-1}} \times \baseQbs{Y^\bot} \times \baseQbs{X_{k+1}} \times \ldots \times \baseQbs{X_n} \to \baseQbs{X_k^\bot}\]
such that for all $x_1 \in \baseQbs{X_1}, \ldots, x_n \in \baseQbs{X_n}, \eta \in \baseQbs{Y^\bot}$,
\[\eta\, f(x_1, \ldots, x_n) = f^{\bot k}(x_1, \ldots, x_{k-1}, \eta, x_{k+1}, \ldots, x_n)\, x_k.\]
When $n=1$, we say that $f$ is \emph{linear}, we write $f^\bot$ for $f^{\bot 1}$, and we write $\eta f$ for $f^\bot(\eta)$, so that the linearity condition reads $(\eta f) x = \eta (f (x))$.
\end{definition}

If $f$ is $n$-linear, then it is clear that $f^{\bot k}$ is $n$-linear for all $1\leq k\leq n$. 

One can easily check that multilinear maps commute with countable sums and expected values with respect to each argument. In particular, all linear maps are morphisms of $S$-algebras. We will see below (Fact \ref{fact:linear-markov}) that in the case of data types, this necessary condition is also sufficient.

\begin{fact}\label{fact:QbsConv-sym-multicat} Convex QBSs and multilinear maps between them form a symmetric multicategory (with composition and symmetries as in the symmetric multicategory $\operatorname{Set}$), which we denote by $\QbsConv$.
\end{fact}
\begin{proof} It is straightforward to check that composition preserves multilinearity. Since composition and the symmetries are inherited from the symmetric multicategory $\operatorname{Set}$, they satisfy the coherence axioms of symmetric multicategories.
\end{proof}
In particular, convex QBSs and linear maps between them form a category, which we also denote by $\QbsConv$. Note that $\placeholder^\bot$ is a functor from $\QbsConv^{\operatorname{op}}$ to $\QbsConv$.

\begin{definition} We make $\emb{\placeholder} : \operatorname{Qbs} \to \QbsConv$ into a map of symmetric multicategories \cite[Definitions 2.1.9 and 2.2.21]{leinster03} by letting $\emb f(\mu_{1},\ldots,\mu_{n})=\int_{x_{1}\in A_{1}}\ldots\int_{x_{n}\in A_{n}}\emb{f(x_{1},\ldots,x_{n})}\,\mu_{n}(\dif x_{n})\ldots\mu_{1}(\dif x_{1})\in\emb B$ for all $f\in\operatorname{Qbs}(A_{1},\ldots, A_{n};B)$ and all $\mu_{1}\in\emb{A_{1}},\ldots,\mu_{n}\in\emb{A_{n}}$.
\end{definition}

In particular, $\emb{\placeholder}$ is a functor from the category $\operatorname{Qbs}$ to the category $\QbsConv$.

\begin{definition} For all convex QBSs $X_1, \ldots, X_n, Y$, we make the set $\QbsConv(X_1, \ldots, X_n; Y)$ of all $n$-linear maps from $\left(X_1, \ldots, X_n\right)$ to $Y$ into a QBS as follows: for all maps $(r \mapsto f_r)$ from $\R$ to $\QbsConv(X_1, \ldots, X_n; Y)$,
$(r \mapsto f_r) \in M_{\QbsConv(X_1, \ldots, X_n; Y)}$ if and only if
\begin{itemize}
\item $(r \mapsto f_r)$ is a measurable as a map from $\R$ to $\operatorname{Qbs}(\baseQbs{X_1} \times \ldots \times \baseQbs{X_n},\allowbreak \baseQbs{Y})$,
\item for all $k$, $(r \mapsto f_r^{\bot k})$ is measurable as a map from $\R$ to $\operatorname{Qbs}(\baseQbs{X_1} \times \ldots \times \baseQbs{X_{k-1}} \times \baseQbs{Y^\bot} \times \baseQbs{X_{k+1}} \times \ldots \times \baseQbs{X_n}, \baseQbs{X_k^\bot})$.
\end{itemize}
\end{definition}

Multilinear maps between data types correspond exactly to sub-probability kernels:

\begin{fact}\label{fact:linear-markov} For all QBSs $A_1, \ldots, A_n, B$, the maps
\[f \mapsto (\mu_1, \ldots, \mu_n) \mapsto \iint_{\scriptscriptstyle x_1 \in A_1, \ldots, x_n \in A_n} \hspace{-1.95cm} f(x_1, \ldots, x_n)\, \mu_1(\dif x_1)\ldots\mu_n(\dif x_n)\]
\[\text{and } g \mapsto (x_1, \ldots, x_n) \mapsto g{\left(\emb{x_1}, \ldots, \emb{x_n}\right)}\]
define inverse natural isomorphisms between the QBSs $\operatorname{Qbs}(A_1, \ldots, A_n; S(B))$ and $\QbsConv\left(\emb{A_1}, \ldots, \emb{A_n}; \emb{B}\right)$.
\end{fact}

In particular, a map $f : \baseQbs{\emb{A}} \to \baseQbs{\emb{B}}$ is linear if and only if it is a morphism of $S$-algebras.

Now, we need to equip spaces of multilinear maps with a structure of convex QBS. In other words, we need to define \emph{linear tests} on multilinear maps. Intuitively, to test a linear map means to apply it to a randomly chosen input and then test its output with a randomly chosen test:

\begin{notation} Let $X_1, \ldots, X_n, Y$ be convex QBSs. For all $\theta = (\mu_p)_{p\in\N}$ in 
$S{\left(\baseQbs{X_1} \times \ldots \times \baseQbs{X_n} \times \baseQbs{Y^\bot}\right)}^\N$
and all $f \in \QbsConv(X_1, \ldots, X_n; Y)$, we write $\operatorname{Test}(f, \theta)$ for
\[\sum_{p\in\N} \int_{\scriptscriptstyle{\baseQbs{X_1} \times \ldots \times \baseQbs{X_n} \times \baseQbs{Y^\bot}}} \hspace{-2.2cm} \eta\,f(x_1,\ldots,x_n)\,\mu_p(\dif (x_1, \ldots, x_n, \eta)).\]

We denote by $\QbsConv^\bot(X_1, \ldots, X_n; Y)$ the quotient of the QBS
\[\left\{ \begin{array}{c}\theta \in S\left(\baseQbs{X_1} \times \ldots \times \baseQbs{X_n} \times \baseQbs{Y^\bot}\right)^\N;\\
\forall f \in \QbsConv(X_1, \ldots, X_n; Y), \operatorname{Test}(f, \theta) \leq 1
\end{array} \right\}\]
by the equivalence relation that identifies $\theta_1$ and $\theta_2$ if and only if $\operatorname{Test}(f, \theta_1) = \operatorname{Test}(f, \theta_2)$ for all $f \in \QbsConv(X_1, \ldots, X_n; Y)$. We denote the equivalence class of $\theta$ by $[\theta]$.
\end{notation}

\begin{definition} Let $X_1, \ldots, X_n, Y$ be convex QBSs. We define a convex QBS $(X_1, \ldots, X_n) \linear Y$ by
\begin{itemize}
\item $\baseQbs{(X_1, \ldots, X_n) \linear Y} = \QbsConv(X_1, \ldots, X_n; Y)$,
\item $\baseQbs{\left((X_1, \ldots, X_n) {\linear} Y\right)^\bot} = \QbsConv^\bot(X_1, \ldots, X_n; Y)$,
\item $[\theta] f = \operatorname{Test}(f, \theta)$.
\end{itemize}

We make $(\placeholder, \ldots, \placeholder) \linear \placeholder$ into a functor from $\left(\QbsConv^{\operatorname{op}}\right)^n \times \QbsConv$ to $\QbsConv$ by letting $\left((\alpha_1, \ldots, \alpha_n) \linear \beta \right) f = (x_1, \ldots, x_n) \mapsto \beta (f (\alpha_1(x_1), \ldots, \alpha_n(x_n)))$.
\end{definition}

For all convex QBSs $X_1, \ldots, X_n, Y$, all permutations $\sigma$ of $[n] = \{1, \ldots, n\}$ and all $f : (X_{\sigma(1)}, \ldots, X_{\sigma(n)}) \linear Y$, we denote by $\sigma^*f$ the linear map $(x_1, \ldots, x_n) \mapsto f(x_{\sigma(1)}, \ldots, x_{\sigma(n)})$. This defines a natural isomorphism $\sigma^*$ between $(X_{\sigma(1)}, \ldots, X_{\sigma(n)}) \linear Y$ and $(X_1, \ldots, X_n) \linear Y$. We write $\sigma_*$ for $\left(\sigma^*\right)^\bot$.

The convex QBS $\W$ is both neutral and dualising:
\begin{fact}\label{fact:W-neutral-dualising} For all convex QBSs $X_1, \ldots, X_n, Y$, the map
\[f \mapsto (x_1, \ldots, x_n, r) \mapsto r\,f(x_1, \ldots, x_n)\]
defines a natural isomorphism between $(X_1, \ldots, X_n) \linear Y$ and $(X_1, \ldots, X_n, \W) \linear Y$, and the map
$\eta \mapsto y \mapsto \eta y$
defines a natural isomorphism between $Y^\bot$ and $Y \linear \W$.
\end{fact}

The symmetric multicategory $\QbsConv$ is closed in the following sense:

\begin{proposition}\label{prop:QbsConv-closed} For all convex QBSs $X_1, \ldots, X_m,\allowbreak Y_1, \ldots, Y_n, Z$, the maps
\[f \mapsto (x_1, \ldots, x_m) {\mapsto} (y_1, \ldots, y_n) {\mapsto} f(x_1, \ldots, x_m, y_1, \ldots, y_n)\]
\[F \mapsto (x_1, \ldots, x_m, y_1, \ldots, y_n) \mapsto F(x_1, \ldots, x_m)(y_1, \ldots, y_m)\]
define inverse natural isomorphisms between $(X_1, \ldots, X_m,\allowbreak Y_1, \ldots, Y_n) \linear Z$ and $(X_1, \ldots, X_m) \linear (Y_1, \ldots, Y_n) \linear Z$.
\end{proposition}
\begin{proof}
For simplicity, we assume $m=n=1$, and we drop the corresponding indices. The general proof is similar.

First, we check that the first map, which we will denote by $\varphi$, is well-defined. Let $f \in (X,Y) \linear Z$. For all $x \in X$, $\varphi(f)(x)$ is linear, with $\varphi(f)(x)^\bot = (\zeta \mapsto f^{\bot 2}(x, \zeta))$. In addition, $\varphi(f)$ is linear, with $\varphi(f)^\bot = ([\theta] \mapsto (x \mapsto \operatorname{Test}(\varphi(f)(x), \theta)))$ (identifying $X^\bot$ with $X \linear \W$): indeed, for all $(y,\zeta) \in \baseQbs{Y}\times\baseQbs{Z^\bot}$ and all $x \in \baseQbs{X}$,  $\operatorname{Test}\left(\varphi(f)(x), \emb{(y,\zeta)}\right) = \zeta\, f(x,y) = f^{\bot 1}(\zeta, y)(x)$.

Now, we check that $\varphi$ is linear. For all $f \in (X,Y) \linear Z$, all $x \in X$ and all $[\theta_2] \in (Y \linear Z)^\bot$, $\operatorname{Test}\left(\varphi(f), \emb{(x, [\theta_2])}\right) = \operatorname{Test}\left(f, \int \emb{(x,y,\zeta)}\, \theta_2(\dif (y, \zeta))\right)$, so $\varphi$ is linear, with $\varphi^\bot = [\theta_1] \mapsto f \mapsto \operatorname{Test}((x, [\theta_2]) \mapsto \operatorname{Test}((y, \zeta) \mapsto \zeta f (x,y), \theta_2), \theta_1)$.

The proof of linearity for the inverse map is similar.
\end{proof}

We saw that whenever $f$ is $n$-linear, $f^{\bot k}$ is also $n$-linear for all $k$. In fact, there is a stronger result:

\begin{fact}\label{fact:bot-iso-maps}
For all convex QBSs $X_1, \ldots, X_n, Y$ and all $k \leq n$, the map
$f \mapsto f^{\bot k}$
defines a natural isomorphism between $(X_1, \ldots, X_n) \linear Y$ and $(X_1, \ldots, X_{k-1}, \allowbreak Y^\bot,\allowbreak X_{k+1}, \ldots, X_n) \linear X_k^\bot$.
\end{fact}

For all convex QBSs $X, Y$, we let $X \otimes Y = \left((X,Y) \linear \W\right)^\bot$. For all $x \in X$ and $y \in Y$, we denote by $x \otimes y$ the unique element of $\baseQbs{X \otimes Y}$ such that $f (x \otimes y) = f(x,y)$ for all $f \in \baseQbs{(X,Y) \linear \W} = \baseQbs{(X\otimes Y)^\bot}$. It follows from the above discussion that the map
\[f \mapsto (w_1, \ldots, w_n, x, y) \mapsto f(w_1, \ldots, w_n, x \otimes y)\]
defines a natural isomorphism between $(W_1, \ldots, W_n,\allowbreak X \otimes Y) \linear Z$ and $(W_1, \ldots, W_n, X, Y) \linear Z$, and therefore that $\left(\QbsConv, \W, \otimes, \linear\right)$ is a closed symmetric monoidal category. For all maps $f : (X,Y) \multimap Z$, we will denote by $x \otimes y \mapsto f(x,y)$ the corresponding map in $X \otimes Y \multimap Z$. 

As in probabilistic coherence spaces, the intuition behind the tensor product $\otimes$ is that a random value of type $X \otimes Y$ is a random pair of values of types $X$ and $Y$, that is to say, two random values of types $X$ and $Y$ that have to be sampled jointly. This intuition is supported by how the tensor product behaves on data types:

\begin{fact} For all QBSs $A$ and $B$, the map
\[\rho \mapsto \int_{(x,y) \in A \times B} \emb{x} {\otimes} \emb{y}\, \rho(\dif(x,y))\]
defines a natural isomorphism between $\emb{A \times B}$ and $\emb{A} \otimes \emb{B}$.
\end{fact}
\begin{proof} Consequence of Fact \ref{fact:linear-markov}.
\end{proof}

\section{Additive connectives}

\begin{definition} Let $(X_i)_{i \in I}$ be a countable family of convex QBSs. We define a convex QBS $\bigwith_{i \in I} X_i$ by
\begin{itemize}
\item $\baseQbs{\bigwith_{i \in I} X_i} = \prod_{i \in I} \baseQbs{X_i}$,
\item $\baseQbs{\left(\bigwith_{i \in I} X_i\right)^\bot} = \{ (\eta_i)_{i\in I} \in \prod_{i \in I} \baseQbs{X_i^\bot};~ \forall (x_i)_{i \in I} \in\prod_{i \in I} \baseQbs{X_i}, \sum_{i \in I} \eta_i\,x_i \leq 1 \}$,
\item $(\eta_i)_{i\in I} \cdot (x_i)_{i \in I} = \sum_{i \in I} \eta_i\,x_i$.
\end{itemize}
In addition, we let $\bigoplus_{i \in I} X_i = \left(\bigwith_{i \in I} X_i^\bot\right)^\bot$.
\end{definition}

For all $j \in I$, we let $\pi_j$ denote the projection from $\prod_{i \in I} \baseQbs{X_i}$ to $X_j$, and $L_j$ the map $x_j \mapsto \left(\begin{array}{ll}x_j & \text{if } i=j \\ 0_{X_i} & \text{otherwise}\end{array}\right)_{i\in I}$ from $X_j$ to $\prod_{i \in I} \baseQbs{X_i}$.

\begin{fact} Let $(X_i)_{i \in I}$ be a countable family of convex QBSs: $\left(\bigwith_{i \in I} X_i, \left(\pi_i\right)_{i \in I}\right)$ is a cartesian product of the family $(X_i)_{i \in I}$, and $\left(\bigoplus_{i \in I} X_i, \left(L_i\right)_{i \in I}\right)$ is a coproduct of the family $(X_i)_{i \in I}$. In addition, $\zero$ is both initial and terminal.
\end{fact}

As in probabilistic coherence spaces, the intuition is that a random value of type $X \with Y$ is in fact two random values of types $X$ and $Y$ that can be sampled separately, while a random value of type $X \oplus Y$  is one that, every time it is sampled, yields either a value of type $X$ or a value of type $Y$. In the case of $\with$, this is just an other way of saying we have a cartesian product, while in the case of $\oplus$, the intuition is supported by how it behaves on data types:

\begin{fact} For all countable families $(A_i)_{i \in I}$ of QBSs, the map
\[[\alpha, \mu] \mapsto \sum_{i\in I}\int_{r \in \alpha^{-1}(A_i)} L_i{\left(\emb{\alpha(r)}\right)}\, \mu(\dif r)\]
defines a natural isomorphism between $\emb{\coprod_{i\in I} A_i}$ and $\bigoplus_{i \in I} \emb{A_i}$.
\end{fact}

\section{Symmetric maps and tensors}

In order to define the exponential modalities, we will need to define the spaces of symmetric maps and symmetric tensors.

For all $n\in\N$, we denote by $\operatorname{Sym}(n)$ the group of permutations of $[n] = \{1, \ldots, n\}$.

\begin{definition} Let $n$ be a natural number and $X, Y$ convex QBSs. We write $(X)^n \linear Y$ for $(X,\ldots, X) \multimap Y$, where $X$ appears $n$ times.\begin{itemize}
\item an $n$-linear map $f : (X)^n \linear Y$ is \emph{symmetric} if $\sigma^*f = f$ for all $\sigma \in \operatorname{Sym}(n)$,
\item a test $\eta \in \left((X)^n \linear Y\right)^\bot$ is \emph{symmetric} if $\sigma_* \eta = \eta$ for all $\sigma \in \operatorname{Sym}(n)$.
\end{itemize}
\end{definition}

If two symmetric maps $f_1, f_2 \in (X)^n \linear Y$ are such that $\eta f_1 = \eta f_2$ for all symmetric tests $\eta$, then $f_1 = f_2$. Indeed, let $\eta$ be any test in $\left((X)^n \linear Y\right)^\bot$: then $\sum_{\sigma \in \operatorname{Sym}(n)} \frac{1}{n!} \sigma_*\eta$ is symmetric, so $\left(\sum_{\sigma} \frac{1}{n!} \sigma_*\eta\right) f_1 = \left(\sum_{\sigma} \frac{1}{n!} \sigma_*\eta\right) f_2$, which means that $\eta \left(\sum_{\sigma} \frac{1}{n!} \sigma^*f_1\right) = \eta \left(\sum_{\sigma} \frac{1}{n!} \sigma^*f_2\right)$; since $f_1$ and $f_2$ are symmetric, this implies $\eta f_1 =  \eta f_2$. Conversely, any two symmetric tests that coincide on symmetric maps are equal. As a result, we can define convex QBSs of \emph{symmetric maps} and \emph{symmetric tensors}:

\begin{definition} Let $n$ be a natural number and $X, Y$ convex QBSs. We define a convex QBS $(X)^n_s \linear Y$ by
\begin{itemize}
\item $\baseQbs{(X)^n_s \linear Y}$ is the set of all symmetric maps in $\baseQbs{(X)^n \linear Y}$,
\item $\baseQbs{\left((X)^n_s \linear Y\right)^\bot}$ is the set of all symmetric tests in $\baseQbs{\left((X)^n \linear Y\right)^\bot}$,
\item $\eta \cdot_{(X)^n_s \linear Y} f = \eta \cdot_{(X)^n \linear Y} f$.
\end{itemize}
We let $X^{\otimes n} = \left((X)^n \linear \W\right)^\bot$ and $X^{\otimes_s n} = \left((X)^n_s \linear \W\right)^\bot$.
\end{definition}

For all $n\in\N$, we define a linear map $\mathcal{S}_n : X^{\otimes n} \linear X^{\otimes_s n}$ by $\mathcal{S}_n = \sum_{\sigma \in \operatorname{Sym}(n)} \frac{1}{n!} \sigma_{*}$. The restriction of $\mathcal{S}_n$ to $X^{\otimes_s n}$ is the identity.

\section{The exponential modalities}

So far, we have defined \emph{linear maps} between QBSs. We would like to define a more general notion of ``computable'' maps (which we will call \emph{analytic maps}, due to their similarity with power series). Following the paradigm of linear logic  \cite{blindspot11}, the first step will be to define the exponential modality ``$!$'' (``of course''). The other exponential modality, ``$?$'' (``why not''), can be defined by duality.

\subsection{Defining $\oc X$}

In this subsection, we fix a convex QBS $X$. We will define the convex QBS $\oc X$ using a generic construction by Melliès, Tabareau and Tasson \cite{bang09}. Following their terminology, we call a \emph{pointed object} any pair $(Y, u)$ with $Y$ a convex QBS and $u : Y \linear \W$, and a \emph{pointed morphism} from $(Y, u)$ to $(Z, v)$ any linear map $f : Y \linear Z$ such that $u = v \circ f$. In order to apply this construction, we only need to prove two conditions. The first is that $(X \with \W, (x,r) \mapsto r, (x,r) \mapsto x)$ defines a \emph{free pointed object} over $X$ in the following sense:

\begin{fact}\label{fact:free-pointed}
For all pointed objects $(Y, u)$ and all linear maps $f : Y \linear X$, there exists a unique pointed morphism $g$ form $(Y, u)$ to $(X \with \W, (x,r) \mapsto r)$ such that $((x,r) \mapsto x) \circ g = f$.
\end{fact}
\begin{proof} Let $g$ be any map from $\baseQbs Y$ to $\baseQbs{X \with \W}$. Then $g$ is a pointed morphism satisfying this hypothesis if and only for all $y \in Y$, $g(y) = (f(y), u(y))$.
\end{proof}

For all $m \in \N$, we denote by $\J_{m+1,m}$ the canonical projection $\left(X \with \W\right)^{\otimes_s m+1} \linear \left(X \with \W\right)^{\otimes_s m}$, that is to say, the restriction to $\left(X \with \W\right)^{\otimes_s m+1}$ of the unique linear map $\left(X \with \W\right)^{\otimes m+1} \linear \left(X \with \W\right)^{\otimes m}$ that sends $(x_1, r_1) \otimes \ldots \otimes (x_{m+1},r_{m+1})$ to $r_{m+1} (x_1, r_1) \otimes \ldots \otimes (x_{m},r_{m})$. For all $n \geq m$, we let $\J_{n,m} = \J_{m+1,m} \circ \ldots \circ \J_{n,n-1}$. The second condition we need to prove is that the diagram
\[\left(X {\with} \W\right)^{\otimes_s 0} \leftarrow \ldots \left(X {\with} \W\right)^{\otimes_s m} \overset{\scriptscriptstyle\J{m+1, m}}{\leftarrow} \left(X {\with} \W\right)^{\otimes_s m+1} \ldots\]
has a limit and that this limit commutes with the tensor product: $\oc X$ will be defined as this limit. The remainder of this section deals with the technical details of how to do this.

The obvious choice for the underlying QBS $\baseQbs{\oc X}$ is the set of all $(a_n)_{n\in\N} \in \baseQbs{\bigwith_{n\in\N} \left(X \with \W\right)^{\otimes_s n}}$ such that $\J_{n,m} (a_n) = a_m$ for all $m \leq n$. In fact, it would be easy to define a structure of convex QBS on top of this, to prove that it is a limit of the above diagram, and to prove that this limit commutes with tensor products \emph{if} the maps $\J_{n,m}$ had sections. However, they are not even necessarily surjective. Indeed, consider for example the case where $X = \W \oplus \W$, $m = 2$ and $n = 3$. Let $f : \left(X \with \W\right)^{\otimes 2} \linear \W$ be defined by $f((r_0,r_1,r_*) \otimes (s_0, s_1, s_*)) = r_0 s_1 + r_1 s_0$. An elementary computation shows that for all $a \in \left(X \with \W\right)^{\otimes 3}$, $f(\J_{3,2}(\mathcal{S}(a))) \leq \frac{2}{3}$, whereas $f\left(\frac{1}{2}(0,1,0)\otimes(1,0,0) + \frac{1}{2}(1,0,0)\otimes(0,1,0)\right) = 1$.

So instead, we will prove that $\J_{n,m}$ has a section \emph{up to a factor that depends only on $m$}. Namely, for all $m \in \N$, we will define a real number $\rho_m \geq 1$, and for all $m,n \in \N$, we will define a linear map $\K_{m,n} : \left(X \with \W\right)^{\otimes_s m} \linear \left(X \with \W\right)^{\otimes_s n}$ such that:\begin{itemize}
\item $\K_{m,n} = \frac{1}{\rho_m} \J_{m,n}$ if $m \geq n$,
\item $\J_{n,m} \circ \K_{m,n} = \frac{1}{\rho_m} \operatorname{id}_{\left(X \with \W\right)^{\otimes_s m}}$ if $m \leq n$.
\end{itemize}

An element of $X^{\otimes n}$ can be seen as a (non-commutative) homogeneous polynomial of degree $n$. Likewise, an element of $\left(X \with \W\right)^{\otimes n}$ can be seen as a (non-necessarily homogeneous) polynomial of degree at most $n$. For example, $(x, 1) \otimes (y, 1) = (x, 0) \otimes (y, 0) + (x, 0) \otimes (0, 1) + (0, 1) \otimes (y, 0) + (0, 1) \otimes (0, 1)$ represents the polynomial ``$x \otimes y + x + y + 1$'' (which we put between quotes because this is not a well-defined notation). 

Homogeneous polynomials can be extracted from elements of $\left(X \with \W\right)^{\otimes_s m}$ as follows: for all $m,n \in \N$, we denote by $\M_{n,m}$ the canonical projection $\left(X \with \W\right)^{\otimes_s n} \linear X^{\otimes_s m}$, that is to say, the restriction to $\left(X \with \W\right)^{\otimes_s n}$ of the unique linear map $\left(X \with \W\right)^{\otimes n} \linear X^{\otimes m}$ that maps $(x_{1},r_{1})\otimes\ldots\otimes(x_{n},r_{n})$ to $r_{m+1} \ldots r_n\, x_1 \otimes \ldots \otimes x_m$ if $m \leq n$, and to $0$ if $m > n$.

For all $m,n\in\N$, we denote by $\inj{m}{n}$ the set of all injections from $[m]$ to $[n]$. We want to define for all $m\leq n$ a section of $\J_{n,m}$ up to a factor that depends only on $m$. To that end, we will first define for all $m \leq n$ a section of $\M_{n,m}$ up to a factor that depends only on $m$:

\begin{lemma} For all $m,p>0$, there exists a unique linear map from $X^{\otimes m}$ to $\left(X \with \W\right)^{\otimes_s mp}$ that maps $x_1 \otimes \ldots \otimes x_m$ to
\[\sum_{j \in \inj{m}{mp}} \frac{1}{m^m} \bigotimes_{k=1}^{mp} \left(\begin{array}{cl}
\left(x_{j^{-1}\left(k\right)},0\right) & \text{if }k\in\operatorname{im}\left(j\right)\\
\left(0,1\right) & \text{if }k\notin\operatorname{im}\left(j\right)
\end{array}\right).\]
\end{lemma}
\begin{proof} For all $\eta \in \baseQbs{\left(\left(X \with \W\right)^{\otimes_s mp}\right)^\bot}$ and all $x_1,\allowbreak \ldots, x_m \in X$,
\[\arraycolsep=1pt \begin{array}{rcl} 1 & \geq & \eta \left((x_1, 1)^{\otimes p} \otimes \ldots \otimes (x_m, 1)^{\otimes p}\right) \\
& \geq & \operatorname{Avg}_{j \in \inj{m}{mp}} p^m \eta\bigotimes_{k=1}^{mp} \left(\begin{array}{cl}
\left(x_{j^{-1}\left(k\right)},0\right) & k\in\operatorname{im}\left(j\right)\\
\left(0,1\right) & k\notin\operatorname{im}\left(j\right)
 \end{array}\right)\end{array}\]
(where $\operatorname{Avg}$ stands for average). Indeed, out of the $2^{mp}$ terms obtained by developing the product $(x_1, 1)^{\otimes p} \otimes \ldots \otimes (x_m, 1)^{\otimes p}$ (where $(y, r)$ is to be read as $(y,0)+(0,r)$), $p^m$ are of the form: one factor $(x_i, 0)$ for each $i \in [m]$, and all the other factors equal to $(0,1)$ (and $\eta$ takes the same value on all such terms, because it is symmetric, so only their number matters). Since there are $\frac{(mp)!}{(mp-m)!}$ injections from $[m]$ to $[mp]$, and since $p^m \frac{(mp-m)!}{(mp)!} \geq \frac{1}{m^m}$, we get
\[1 \geq \sum_{j \in \inj{m}{mp}} \frac{1}{m^m} \eta \bigotimes_{k=1}^{mp} \arraycolsep=1pt \left(\begin{array}{cl}
\left(x_{j^{-1}\left(k\right)},0\right) & \text{if }k\in\operatorname{im}\left(j\right)\\
\left(0,1\right) & \text{if }k\notin\operatorname{im}\left(j\right)
\end{array}\right).\]
\end{proof}

We denote by $\NN_{m,mp}$ the restriction of this map to $X^{\otimes_s m}$: one can check that  $\M_{mp,m} \circ \NN_{m,mp} = \frac{m!}{m^m} \operatorname{id}_{X^{\otimes_s m}}$. For all $m>0$ and all $n\in\N$, we define $\NN_{m,n} : X^{\otimes_s m} \linear \left(X \with \W\right)^{\otimes_s n}$ by $\NN_{m,n} = \J_{mp,n} \circ \NN_{m,mp}$, where $p$ is smallest positive integer such that $mp \geq n$ (which makes sense because $\J_{mp,n}$ is the identity when $n=mp$). Finally we define $\NN_{0,n} : X^{\otimes_s 0} \linear \left(X \with \W\right)^{\otimes_s n}$ for all $n\in\N$ by $\NN_{0,n}(r) = r\bigotimes_{k=1}^{n}(0,1)$. One can check that for all $m,n \in \N$, 
\[\M_{n,m} \circ \NN_{m,n} = \left\{ \begin{array}{cl}
\frac{m!}{m^m} \operatorname{id}_{X^{\otimes_s m}} & \text{if } n\geq m \\
0 & \text{if } n<m
\end{array} \right.\] (with the convention that $0^0 = 1$). Intuitively, $\NN_{m,n}$ takes a homogeneous polynomial of degree $m$ and, if possible, represents it as an element of $\left(X \with \W\right)^{\otimes_s n}$, up to a factor $\frac{m!}{m^m}$.

For all $m \in \N$, we let $\rho_m = \frac{m^m(m+1)}{m!}$. For all $m,n \in \N$, we define $\K_{m,n} : \left(X \with \W\right)^{\otimes_s m} \linear \left(X \with \W\right)^{\otimes_s n}$ by
\[ \K_{m,n} = \sum_{k=0}^m \frac{1}{(m+1)}\frac{k^k\, m!}{k!\, m^m}  \NN_{k,n} \circ \M_{m,k}\]
(which is well-defined because $\frac{k^k}{k!} \leq \frac{m^m}{m!}$ for all $k \leq m$). In other words, for all $k$, $\K_{m,n}$ extracts from its argument the homogeneous part of degree $k$, turns that part into an element of $\left(X \with \W\right)^{\otimes_s n}$ up to a factor $\frac{1}{\rho_m}$, and then sums all the results. Thus, for all $m,n \in \N$, we do have:\begin{itemize}
\item $\K_{m,n} = \frac{1}{\rho_m} \J_{m,n}$ if $m \geq n$,
\item $\J_{n,m} \circ \K_{m,n} = \frac{1}{\rho_m} \operatorname{id}_{\left(X \with \W\right)^{\otimes_s m}}$ if $m \leq n$.
\end{itemize}

All this means that the following definition makes sense:

\begin{definition}
We define a convex QBS $\oc X$ as follows
\begin{itemize}
\item $\baseQbs{\oc X}=\left\{ \begin{array}{l}
(a_{n})_{n\in\N}\in\prod_{n\in\N}(X\with\W)^{\otimes_{s}n};\\
\forall m\leq n,a_{m}=\J_{n,m}(a_{n})
\end{array}\right\} $ 
\item $\baseQbs{(\oc X)^{\bot}}$ is the set of all families of maps $(f_{n})_{n\in\N}\in\prod_{n\in\N}\operatorname{Qbs}\left(\baseQbs{(X\with\W)^{\otimes n}},[0,+\infty)\right)$ such that\begin{itemize}
\item for all $n$, $\frac{f_{n}}{\rho_{n}}$ is in $(X \with \W)^{\otimes n} \multimap \W$ and is symmetric,
\item for all $m \leq n$, $f_m = \rho_m\, f_n \circ \K_{m,n}$,
\item for all $(a_n)_{n\in\N} \in \baseQbs{\oc X}$, $\sup_{n \in \N} f_n(a_n) \leq 1$,
\end{itemize}
with the subset QBS structure,
\item $(f_{n})_{n\in\N} \cdot_{\oc X} (a_{n})_{n\in\N} = \sup_{n \in \N} f_n(a_n)$.
\end{itemize}
\end{definition}

As stated in the introduction, to make this definition usable, we need to reformulate it in terms of countable sums. To this end, for all $n,k \in \N$, we define a linear map $\D_{n,k} : (X \with \W)^{\otimes_s n} \linear (X \with \W)^{\otimes_s n}$ by $\D_{n,k} = \frac{k^k}{k!}\, \NN_{k,n} \circ \M_{n,k}$. This map extracts the homogeneous part of degree $k$ without changing the type of its argument. In particular, $\operatorname{id}_{(X \with \W)^{\otimes_s n}} = \sum_{k=0}^n \D_{n,k}$.

\begin{fact}\label{fact:bang-charac-sums} For all families of maps $(f_{n})_{n\in\N}\in\prod_{n\in\N}\operatorname{Qbs}\left(\baseQbs{(X\with\W)^{\otimes n}},[0,+\infty)\right)$,
$(f_{n})_{n\in\N} \in \baseQbs{(\oc X)^{\bot}}$ if and only if
\begin{itemize}
\item for all $n$, $\frac{f_{n}}{\rho_{n}}$ is in $(X \with \W)^{\otimes n} \multimap \W$ and is symmetric,
\item for all $n$, $f_{n+1} = f_{n}\circ\J_{n+1,n}+f_{n+1}\circ\D_{n+1,n+1}$,
\item for all $(a_n)_{n\in\N} \in \baseQbs{\oc X}$, $\sum_{n\in\mathbb{N}}f_{n}\circ\D_{n,n}\left(a_{n}\right)\leq1$.
\end{itemize}
In addition, for all $(a_n)_{n\in\N} \in \baseQbs{\oc X}$, $(f_{n})_{n\in\N} \cdot_{\oc X} (a_{n})_{n\in\N} = \sum_{n\in\mathbb{N}}f_{n}\circ\D_{n,n}\left(a_{n}\right)$.
\end{fact}

For all $n \in \N$, we denote by $\pi_n$ the canonical projection $\oc X \linear  (X \with \W)^{\otimes_s n}$. These projections also have sections up to a factor $\frac{1}{\rho_n}$:

\begin{notation} Let $n \in \N$. For all $a_n \in (X \with \W)^{\otimes_s n}$, we let $\theta_n(a_n) = \left(\K_{n,m}a_n\right)_{m\in\N}$. This defines a linear map $\theta_n : (X \with \W)^{\otimes_s n} \linear \oc X$ that satisfies the equation $\pi_n \circ \theta_n = \frac{1}{\rho_n} \operatorname{id}_{(X \with \W)^{\otimes_s n}}$.
\end{notation}

With that, it is clear that $\oc X$ is the limit of the diagram
\[\left(X {\with} \W\right)^{\otimes_s 0} \leftarrow \ldots \left(X {\with} \W\right)^{\otimes_s m} \overset{\J{m+1, m}}{\leftarrow} \left(X {\with} \W\right)^{\otimes_s m+1} \ldots\]
and that this limit commutes with the tensor product, namely:

\begin{theorem}\label{thm:bang-limit} 
Let Z,Y be convex QBSs, and let $\left(\varphi_{n}\right)_{n\in\mathbb{N}}\in\prod_{n\in\mathbb{N}}\left(Z\multimap Y\otimes\left(X\with\W\right)^{\otimes_{s}n}\right)$ be such that for all $m\leq n$, $\left(\operatorname{id}_Y\otimes\mathcal{J}_{n,m}\right)\circ\varphi_{n}=\varphi_{m}$. Then the map
\[\varphi_{\infty}:\left\{ \arraycolsep=1pt \begin{array}{ccc}
\baseQbs{Z} & \to & \baseQbs{Y\otimes\oc X}\\
z & \mapsto & \sup_{n\in\mathbb{N}}\rho_{n}\left(\operatorname{id}_Y\otimes\theta_{n}\right)\circ\varphi_{n}(z)\\
 &  & =\sum_{n\in\mathbb{N}}\rho_{n}\left(\operatorname{id}_Y\otimes(\theta_{n}\circ\mathcal{D}_{n,n})\right)\circ\varphi_{n}(z)
\end{array}\right.\]
\begin{itemize}
\item is well-defined,
\item is a linear map from $Z$ to $Y\otimes\oc X$,
\item is the only map from $\baseQbs{Z}$ to $\baseQbs{Y\otimes\oc X}$ such that for all $n\in\N$, $\left(\operatorname{id}_Y\otimes\pi_{n}\right)\circ\varphi_{\infty}=\varphi_{n}$.
\end{itemize}
\end{theorem}

As in probabilistic coherence spaces, a random value of $\oc X$ represents a generator of random values of $X$ whose distribution is itself random.

\subsection{The free commutative comonoid structure on $\oc X$}

Because of Fact \ref{fact:free-pointed} and Theorem \ref{thm:bang-limit}, we know that for all $X$, $\oc X$ can be equipped with a structure of commutative comonoid freely generated by $X$ \cite[definition in the introduction]{bang09}. In this subsection, we simply spell out this structure and a few constructions that come from it. This will come in handy when defining and proving statements about analytic maps.

First, note that Theorem \ref{thm:bang-limit} turns ``$!$'' into a functor from $\QbsConv$ to $\QbsConv$, with $(\oc f)(a_n)_{n\in\N} = ((f \otimes \operatorname{id}_\W)^{\otimes_s n}(a_n))_{n\in\N}$ for all $f : X \linear Y$ and all $(a_n)_{n\in\N} \in \oc X$.

The following notation will be useful to define linear maps from spaces of the form $\oc X$:

\begin{notation} Let $X,Y$ be convex QBSs, $n\in\N$ and $f : (X \with \W)^n \linear Y$. We denote by
\[\left(\bigotimes_{k=1}^m(x_{m,k},r_{m,k})\right)_{m\in\N} \mapsto f\left((x_{n,1},r_{n,1}),\ldots,(x_{n,n},r_{n,n})\right)\]
the linear map $f \circ \pi_n \in \oc X \linear Y$.
\end{notation}

\begin{definition} For all convex QBSs $X$, using the above notation, we define a linear map $\operatorname{weak}_X : \oc X \linear \W$ as
\[\left(\bigotimes_{k=1}^m(x_{m,k},r_{m,k})\right)_{m\in\N} \mapsto 1,\]
a linear map $\operatorname{cont}_X : \oc X \linear \oc X \otimes \oc X$ as
\[\begin{array}{rl}
& \left(\bigotimes_{k=1}^m(x_{m,k},r_{m,k})\right)_{m\in\N} \\
\mapsto & \sup_{p,q\in\mathbb{N}}\left(\arraycolsep=1pt\begin{array}{c}\left(\rho_{p}\theta_{p}\bigotimes_{k=1}^{p}\left(x_{p+q,k},r_{p+q,k}\right)\right)\\\otimes\left(\rho_{q}\theta_{q}\bigotimes_{k=p+1}^{p+q}\left(x_{p+q,k},r_{p+q,k}\right)\right)\end{array}\right)\\
= & \sum_{p,q\in\mathbb{N}}\left(\arraycolsep=1pt\begin{array}{c}\left(\rho_{p}\theta_{p}\mathcal{D}_{p,p}\bigotimes_{k=1}^{p}\left(x_{p+q,k},r_{p+q,k}\right)\right)\\{\otimes}\left(\rho_{q}\theta_{q}\mathcal{D}_{q,q}\bigotimes_{k=p+1}^{p+q}\left(x_{p+q,k},r_{p+q,k}\right)\right)\end{array}\right)
\end{array}\]
(where composition and application are noted multiplicatively), and a linear map $\operatorname{der}_X : \oc X \linear X$ as
\[\left(\bigotimes_{k=1}^{m}\left(x_{m,k},r_{m,k}\right)\right)_{m\in\mathbb{N}} \mapsto x_{1,1}.\] 
\end{definition}

As a consequence of Theorem \ref{thm:bang-limit}, we get:
 \begin{proposition}\label{prop:free-comonoid}
For all convex QBSs $X$,  $\left(\oc X,\operatorname{weak}_{X},\operatorname{cont}_{X}\right)$ is a commutative comonoid freely generated by $\left(X,\operatorname{der}_{X}\right)$.
 \end{proposition}
 
This is known \cite{lafont88} to imply the following results:
\begin{itemize}
\item Let $\operatorname{dig}_X : \oc X \linear \oc{\oc X}$ be the unique morphism of comonoids such that $\operatorname{der}_{\oc X} \circ \operatorname{dig}_X = \operatorname{id}_{\oc X}$. Then $(!, \operatorname{der}, \operatorname{dig})$ is a comonad. We denote by $\ll$ its Kleisli composition, \textit{i.e.} for all $f\in\oc X\multimap Y$ and all $g\in\oc Y \linear Z$, $g\ll f=g\circ\oc f\circ\operatorname{dig}_{X}$.
\item Let $\operatorname{store}_{X_1,\ldots,X_n} : \oc X_1 \otimes \ldots \otimes \oc X_n \linear \oc {(X_1 \with \ldots \with X_n)}$ be the unique morphism of comonoids such that $\operatorname{der}_{X_1 \with \ldots \with X_n} \circ \operatorname{store}_{X_1,\ldots,X_n} = a_1 \otimes \ldots \otimes a_n \mapsto (\operatorname{der}_{X_1}(a_1), \ldots, \operatorname{der}_{X_n}(a_n))$. Then $\operatorname{store}_{X_1,\ldots,X_n}$ is an isomorphism.
\item The Kleisli category of the comonad $!$ is cartesian closed.
\end{itemize}
 
For all convex QBSs $X$ and all $n\in\mathbb{N}$, we denote by $\operatorname{cont}_{X,n}$ the canonical linear map from $\oc X$ to $\left(\oc X\right)^{\otimes_{s}n}$. In particular, $\operatorname{cont}_{X,2} = \operatorname{cont}_{X}$, $\operatorname{cont}_{X,1} = \operatorname{id}_{\oc X}$, and $\operatorname{cont}_{X,0} = \operatorname{weak}_{X}$.
For all $n\in\mathbb{N}$ and all $m_{1},\ldots,m_{n}\in\mathbb{N}$, one can check that
\begin{itemize}
\item $\left(\pi_{m_{1}}\otimes\ldots\otimes\pi_{m_{n}}\right)\circ\operatorname{cont}_{X,n}=\pi_{m_{1}+\ldots+m_{n}}$,
\item $\left(\operatorname{cont}_{X,m_{1}}\otimes\ldots\otimes\operatorname{cont}_{X,m_{n}}\right)\circ\operatorname{cont}_{X,n}=\operatorname{cont}_{X,m_{1}+\ldots+m_{n}}$ (which is just an other way of saying that $\oc X$ is a comonoid).
\end{itemize}

\section{Analytic maps}

In coherence spaces \cite{blindspot11}, seen as a model of computation, computable functions are represented by \emph{stable maps}. Each coherence space $X$ comes with a universal stable map from $X$ to $\oc X$, in the sense that a map from $X$ to $Y$ is stable if and only if it can be obtained by composing this universal map with a (necessarily unique) linear map from $\oc X$ to $Y$. We use this idea to define analytic maps between convex QBSs, and we prove that convex QBSs and analytic maps form a cartesian closed category.

\begin{definition} Let $X$ be a convex QBS. We define a measurable map $\nabla_X : \baseQbs{X} \to \baseQbs{\oc X}$ by $\nabla_X(x) = \left((x,1)^{\otimes n}\right)_{n\in\N}$.
\end{definition}

This map is injective (because $\operatorname{der}_X \circ \nabla_X = \operatorname{id}_X$), and monotone. A useful remark is that its image is exactly the set of ``co-idempotent'' elements of the comonoid $\oc X$ (minus $0_{\oc X}$):

\begin{fact}\label{fact:nabla-idempotents} Let $X$ be a convex QBS. 
For all $a\in\oc X$, $\operatorname{cont}_{X}(a)=a\otimes a$ if and only if $a=0$ or $a=\nabla_X(\operatorname{der}_X(a))$.
\end{fact}
\begin{proof} One can check from the definition of $\operatorname{cont}$ that $\operatorname{cont}(0)= 0 \otimes 0$ and that $\operatorname{cont}(\nabla(x)) = \nabla(x) \otimes \nabla(x)$ for all $x \in X$.

Assume $\operatorname{cont}(a) = a \otimes a$. Then $\pi_0(a) = (\pi_0 \otimes \pi_0)(a) = \pi_0(a)^2$, therefore $\pi_0(a) = 0$ or $\pi_0(a) = 1$. In addition, for all $n$, $\pi_{n+1}(a) = (\pi_1 \otimes \pi_n) \circ \operatorname{cont} (a) = (\pi_1 \otimes \pi_n)(a\otimes a)$. One can check that $\pi_1(a) = (\operatorname{der}(a), \pi_0(a))$, so by induction, for all $n>0$, $\pi_n(a) = \pi_1(a)^{\otimes n} =(\operatorname{der}(a), \pi_0(a))^{\otimes n}$. If $\pi_0(a) = 1$, that means $a = \nabla(\operatorname{der}(a))$. If $\pi_0(a) = 0$, that means $\pi_{1}(a)=\J_{2,1}(\pi_{2}(a))=\J_{2,1}((\operatorname{der}(a),0)^{\otimes2})=0$, so $a = 0$.
\end{proof}

The map $\nabla_X$ duplicates its input, including side-effects (\textit{i.e.} probabilistic choices and non-termination), and as such it is not linear (unless $X$ is $\zero$). However, as one would expect, in the case of data types, values (not side effects) can be duplicated linearly. Namely, for all QBSs $A$, there exists a unique linear map $\operatorname{copy}_A : \emb A \linear \oc {\emb A}$ such that for all $x \in A$, $\operatorname{copy}_A(\emb x) = \nabla_{\emb A}(\emb x)$ (defined by $\operatorname{copy}_A(\mu) = \int_{x\in A} \nabla_{\emb A}(\emb x)\,\mu(\dif x)$).

\begin{definition} Let $X, Y$ be convex QBSs. An \emph{analytic map} from $X$ to $Y$ is a (necessarily measurable and monotone) map $f : \baseQbs{X} \to \baseQbs{Y}$ such that there exists a linear map $f_! : \oc X \linear Y$ such that $f = f_! \circ \nabla_X$.
\end{definition}

The map $\nabla_X$ is analytic by definition, and it is universal in the following sense:

\begin{theorem}\label{thm:analytic-unique}
 Let $X, Y$ be convex QBSs and $f$ an analytic map from $X$ to $Y$. There exists a unique linear map $f_! : \oc X \linear Y$ such that $f = f_! \circ \nabla_X$.
\end{theorem}
\begin{proof} It is sufficient to prove that for all $\alpha = (\alpha_n)_{n\in\N}, \beta = (\beta_n)_{n\in\N} \in (\oc X)^\bot$ , if $\alpha \nabla(x) = \beta \nabla(x)$ for all $x \in X$, then $\alpha = \beta$. For all $x$, by Fact \ref{fact:bang-charac-sums}, $\alpha \nabla (x) = \sum_{n\in\N} \alpha_n \circ  \D_{n,n} \circ \pi_n \circ \nabla (x) = \sum_{n\in\N} \alpha_n {\left((x,0)^{\otimes n}\right)}$, and similarly for $\beta$. In addition, also by Fact \ref{fact:bang-charac-sums}, in order to prove that $\alpha = \beta$, it is sufficient to prove that for all $n\in\N$, $\alpha_n \circ \D_{n,n} = \beta_n \circ \D_{n,n}$. Therefore, it is sufficient to prove that for all $n\in\N$ and all $x_1, \ldots, x_n \in X$, $\alpha_n((x_1, 0) \otimes \ldots \otimes (x_n, 0)) = \beta_n((x_1, 0) \otimes \ldots \otimes (x_n, 0))$. For all $r_1, \ldots, r_n \geq 0$ such that $r_1 + \ldots + r_n \leq 1$, $\alpha \nabla(r_1 x_1+ \ldots + r_n x_n)$ is the sum of an $n$-variate power series in $r_1, \ldots, r_n$ in which the coefficient of the monomial $r_1 \ldots r_n$ is equal to $n!\,\alpha_n((x_1, 0) \otimes \ldots \otimes (x_n, 0))$. The same can be said about $\beta$, and two $n$-variate power series that coincide on a subset of $\R^n$ with non-empty interior have the same coefficients, therefore $\alpha_n((x_1, 0) \otimes \ldots \otimes (x_n, 0)) = \beta_n((x_1, 0) \otimes \ldots \otimes (x_n, 0))$.

\end{proof}

This allows us to define a structure of convex QBS on the set of analytic maps from $X$ to $Y$, by simply transporting the structure of $\oc X \linear Y$.

\begin{definition} Let $X, Y$ be convex QBSs. We define a convex QBS $X \analytic Y$ as follows:
\begin{itemize}
\item the underlying set of the QBS $\baseQbs{X \analytic Y}$ is the set of all analytic maps from $X$ to $Y$,
\item $M_{\baseQbs{X \analytic Y}} = \{ r \mapsto f_r \circ \nabla_X;~ (r \mapsto f_r) \in M_{\baseQbs{\oc X \linear Y}} \}$,
\item $\baseQbs{(X \analytic Y)^\bot} = \baseQbs{(\oc X \linear Y)^\bot}$,
\item $\eta \cdot_{X \analytic Y} f = \eta \cdot_{\oc X \linear Y} f_!$,
\end{itemize}
\end{definition}

Since $\operatorname{der}_X \circ \nabla_X = \operatorname{id}_X$, for all $f : X \linear Y$, $f$ is analytic and $f_! = f \circ \operatorname{der}_X$.

As a consequence of Fact \ref{fact:nabla-idempotents}, it is easy to prove that $\operatorname{dig}_X \circ \nabla_X = \nabla_{\oc X} \circ \nabla_X$, and that for all $f : X \analytic Y$, $\oc{f_{!}}\circ\nabla_{\oc X}=\nabla_{Y}\circ f_{!}$. As a result:
\begin{proposition}\label{prop:composition-analytic} Let $X, Y, Z$ be convex QBSs. For all analytic maps $f : X \analytic Y$ and $g : Y \analytic Z$, $g \circ f$ is analytic, and $(g \circ f)_! = (g_! \ll f_!)$.
\end{proposition}

This means that convex QBSs and analytic maps form a category that is equivalent to the Kleisli category of ``$!$''. In particular, it is cartesian closed, with $\with$ as a cartesian product and $\analytic$ as an internal hom functor.

There is one last point to check in order to ensure that analytic functions are well-behaved:

\begin{proposition} Let $X, Y$ be convex QBSs. The map $! : \baseQbs{X \linear Y} \to \baseQbs{\oc X \linear \oc Y}$ is analytic.
\end{proposition}
\begin{proof}
For all $n \in \N$, we define $\varphi_n : ((X \linear Y) \with \W)^{\otimes n} \linear  (X \with \W)^{\otimes n}  \linear  (Y \with \W)^{\otimes n}$ by $\varphi_n((f_1, r_1) \otimes \ldots \otimes (f_n, r_n))((x_1, s_1) \otimes \ldots \otimes (x_n, s_n)) = (f_1(x_1), r_1 s_1) \otimes \ldots \otimes (f_n(x_n), r_n s_n)$.

Then for all $f : X \linear Y$ and all $(a_n)_{n\in\N} \in \oc X$, $(\oc f)((a_n)_{n\in\N}) = ((f \with \operatorname{id}_\W)^{\otimes n} (a_n))_{n\in\N} = (\varphi_n((f,1)^{\otimes n})(a_n))_{n\in\N} = (\varphi_n(\pi_n(\nabla(f)))(a_n))_{n\in\N}$.
\end{proof}

Finally, the connection with power series is given by the following result:

\begin{fact} 
\label{fact:analytic-power-series} Let $X, Y$ be convex QBSs. For all $f : X \analytic Y$, there exists a unique family $\left(\frac{\partial_n f}{\rho_n}: (X)_s^n \linear Y\right)_{n\in\N}$ such that for all $x\in X$, $f(x) = \sum_{n\in\N} \rho_n \frac{\partial_n f}{\rho_n} (x, \ldots, x)$.
\end{fact}

\begin{proof} One can check that both points hold if and only if $\frac{\partial_n f}{\rho_n} = \rho_n \, f_! \circ \theta_n \circ \NN_{n,n} \circ \mathcal{S}_n$ for all $n$.
\end{proof}
\begin{corollary}
Let $A$ be a QBS and $Y$ a convex QBS. For all $f,g : \emb A \analytic Y$, if $f$ and $g$ coincide on finitely-supported measures, then $f = g$.
\end{corollary}
\begin{proof} Similar to Theorem \ref{thm:analytic-unique}.
\end{proof}

As in probabilistic coherence spaces, there are non-effective analytic maps: take for example $X = \W \oplus \W$, $Y = \W$, and $f = (p_0, p_1) \mapsto 4 p_0 p_1$. Let $g : \left(X \with \W\right)^2 \linear \W$ be defined by $g((p_0, p_1,p_*), (q_0, q_1,q_*)) = p_0 q_1 + q_0 p_1$. One can check that for all $n \geq 1$, $\norm{g \circ \J_{2n,2}} \leq \frac{n}{2n-1}$, so $\norm{g \circ \pi_2} \leq \frac{1}{2}$ (with $\pi_2 : \oc X \linear \left(X \with \W\right)^{\otimes_s 2}$). As a result, $f = 2g \circ \pi_2 \circ \nabla_X \in X \analytic Y$, yet it is clear that $f$ is not effective as a map from $S(\{0,1\})$ to $S(\{0\})$.

\section{Least fixed points}
\label{section:fixpoints}

We prove that all analytic maps from a convex QBS to itself have a least fixed point, so this denotational model can interpret recursive programs.

\begin{theorem}\label{thm:fixpoints} Let $X$ be a convex QBS. For all analytic maps $f\in X\analytic X$,
\begin{itemize}
\item $f$ has a least fixed point $\operatorname{fix}(f)\in X$,
\item $\operatorname{fix}(f)=\sup_{n\in\N}f^{n}(0)$.
\end{itemize}
Moreover, the map $f\mapsto\operatorname{fix}(f)$ is analytic (that is to say, it is in $(X\analytic X)\analytic X$).
\end{theorem}

\begin{proof}
To get lighter notations, we will omit the ``$\circ$'' when composing linear maps (as is traditional in linear algebra).

The idea behind the proof is to express $\sup_{n\in\N}f^{n}(0)$ as a countable sum.
First, we define a linear map
\[
\varphi:\left\{ \begin{array}{ccc}
\oc X & \multimap & \oc X\\
\left(a_{n}\right)_{n\in\mathbb{N}} & \mapsto & \left(\sum_{k>0}^{k\leq n}\mathcal{D}_{n,k}a_{n}\right)_{n\in\mathbb{N}}
\end{array}\right.,
\]
so that for all $a\in \oc X$, $a=\varphi(a)+\pi_{0}(a)\nabla(0)$.
To make the following reasoning clearer, we will write $a-\pi_{0}(a)\nabla(0)$
for $\varphi(a)$.

For all $g\in X\analytic X$, since $\pi_{0}\left(\oc g_{!}\,\operatorname{dig}\left(\nabla(0)\right)\right)
=\pi_{0}\left(\nabla(g(0))\right)=1$, we write $\oc g_{!}\,\operatorname{dig}\left(\nabla(0)\right)-\nabla(0)$
for $\varphi\left(\oc g_{!}\,\operatorname{dig}\left(\nabla(0)\right)\right)$.

For all $n>0$, we define an analytic map $\beta_n : \left(X\analytic X\right) \analytic X$ as
\[
g \mapsto \operatorname{der}\left(\oc{g_{!}}\,\operatorname{dig}\right)^{n-1}\left(\oc g_{!}\,\operatorname{dig}\left(\nabla(0)\right)-\nabla(0)\right),
\]
so that for all $n\in\mathbb{N}$ and all $g\in X\analytic X$, $g^{n}(0)=\sum_{k=1}^{n}\beta_{k}(g)$.

For all $\eta\in X^{\bot}$ and all $n\in\mathbb{N}$, $\sum_{k=1}^{n}\eta\beta_{k}\left(f\right)=\eta f^{n}(0)\leq1$,
so $\sum_{k\geq1}\beta_{k}(f)\in X$ is well-defined: we
denote it by $\operatorname{fix}(f)$.

One can check that for all $(x_n)_{n\in\N}\in X^\N$ and all $\xi \in (\oc X)^\bot$, $\xi\, \nabla{\left(\sum_{n\in\N}x_n\right)}=\sup_{n\in\N} \xi\, \nabla{\left(\sum_{k\leq n}x_k\right)}$. Therefore, for all $\eta\in X^{\bot}$, 
\[
\begin{array}{rl}
 & \eta f(\operatorname{fix}(f))
= \left(\eta f_{!}\right)\left(\nabla\left(\operatorname{fix} (f)\right)\right)\\
= & \left(\eta f_{!}\right)\left(\nabla\left(\sum_{n\geq1}\beta_{n}(f)\right)\right)\\
= & \sup_{n\in\mathbb{N}}\left(\eta f_{!}\right)\left(\nabla\left(\sum_{k\geq1}^{k\leq n}\beta_{k}(f)\right)\right)\\
= & \sup_{n\in\mathbb{N}}\left(\eta f_{!}\right)\left(\nabla\left(f^{n}(0)\right)\right)\\
= & \sup_{n\in\mathbb{N}}\eta f^{n+1}(0)
= \sum_{n\geq1} \eta \beta_n(f) \\
= & \eta\left(\operatorname{fix}(f)\right),
\end{array}
\]
so $f\left(\operatorname{fix}(f)\right)=\operatorname{fix}(f)$. Moreover, $\operatorname{fix} \in (X \analytic X) \analytic X$, with $\operatorname{fix}_! = \sum_{n\geq1} {\beta_n}_!$.
\end{proof}

\section{A toy probabilistic language} \label{section:example-language}

As an example, we briefly describe a language for which the category of convex QBSs and analytic maps provides an extensional denotational semantics. It can be described as call-by-name PCF with a type for real numbers, primitives for randomly generated reals, a construction to force call-by-value evaluation on data types, and a conditional branching instruction. The types of this language are defined by:
\[A, B := \R ~\vert~ A \times B ~\vert~ A + B ~\vert~ A \to B.\]

We call types written without ``$\to$'' \emph{data types}. Terms are defined by:
\[\begin{array}{rll} t, u, v := && x ~\vert~ \lambda x.\, t ~\vert~ tu ~\vert~ \operatorname{fix } x.\, t ~\vert~ (t,u) ~\vert~ Lt ~\vert~ Ru \\
& \vert & \kw{match} t \kw{with} (x,y) \mapsto u \\
& \vert & \kw{match} t \kw{with} Lx \mapsto u \,; Ry \mapsto v\\
& \vert & f(t_1, \ldots, t_n) \quad (f \in \operatorname{Qbs}(\R^n, S(\R))) \\
& \vert & \kw{if} t \kw{then} u \kw{else} v\\
& \vert & \kw{eval} t \kw{as} x \kw{in} u,\\
\end{array}\]
where $\operatorname{eval}$ is used on data types to force evaluation, and $\operatorname{if}$ tests whether a real number is non-zero. The only typing rule that requires attention is that of $\operatorname{eval}$. The rule is: for all data types $D$,
\[\begin{prooftree} \hypo{\Gamma \vdash t : D} \hypo{\Gamma, x: D \vdash u : A} \infer2{\Gamma \vdash \kw{eval} t \kw{as} x \kw{in}~ u : A} \end{prooftree}.\]

To each type $A$, we associate a convex QBS $\denot{A}$, following Girard's call-by-name translation \cite{girardLL87}:
\[\begin{array}{ll} \denot{\R} = \emb{\R} & \denot{A\times B} = \denot{A} \with \denot{B} \\
\denot{A \to B} = \denot{A} \analytic \denot{B} & \denot{A + B} = \oc{\denot{A}} \oplus \oc{\denot{B}}.\end{array}\]

We interpret each valid typing judgement $x_1 : A_1, \ldots, x_n : A_n \vdash t : B$ by an analytic map $\denot{t} : \denot{A_1} \with \ldots \with \denot{A_n} \analytic \denot{B}$. As with types, we follow Girard's translation, and the only construction that requires attention is $\operatorname{eval}$.
Intuitively, the program ``$\kw{eval} t \kw{as} x \kw{in} u$'' samples $t$ exactly once and then copies the resulting data as many times as needed by $u$. In order to interpret it, we define for all data types $D$ a linear map $\operatorname{copyData}_{D} : \denot{D} \linear \oc{\denot{D}}$ that represents this copying operation. Namely, we let $\operatorname{copyData}_{\R} = \operatorname{copy}_{\R} : \emb\R \linear \oc{\emb\R}$; $\operatorname{copyData}_{D_1+D_2} \circ L_j = \oc{L_j} \allowbreak \circ \operatorname{dig}_{\denot{D_j}} \allowbreak \circ \operatorname{copyData}_{D_j} \circ \allowbreak \operatorname{der}_{\denot{D_j}}$ for $j \in \{1,2\}$; and $\operatorname{copyData}_{D_1 \times D_2} = \operatorname{store}_{\denot{D_1},\denot{D_2}} \circ ((a,b) \mapsto a \otimes b) \circ (\operatorname{copyData}_{D_1} \with \operatorname{copyData}_{D_2})$. Then we let
\[\denot{\kw{eval} t \kw{as} x \kw{in} u} = \denot{u}_! \circ \operatorname{copyData}_{D} \circ \denot{t}\]
(assuming for simplicity that $u$ has no free variable but $x$ -- the general expression is similar but more cumbersome).

Since this is a call-by-name calculus, each member of a pair is sampled independently. If we want a pair whose members are correlated, we need to wrap it inside a constructor. For example, the program ``$\kw{eval} \operatorname{uniform}(\emb{0},\emb{1}) \kw{as} x \kw{in} L(x,x)$'' always produces pairs with identical members: its denotation is $L_1 (\int_0^1 \nabla(\emb x, \emb x) \lambda(\dif x))$, where $\lambda$ denotes the uniform measure.

We could just as well have chosen to interpret a call-by-value calculus, using Girard's call-by-value translation. In fact, the whole discussion so far suggests that probabilistic languages might benefit from linear typing, which makes it possible to mix features from both styles: this is what convex QBSs would be best-suited to interpret.

\section{Conclusion}

We described a model of probabilistic programming (in the narrow sense) that is not limited to discrete probabilities, is compatible with integration, interprets all the connectives of linear logic, and in which all functions have a least fixed point.

A clear direction for future research is to investigate convex QBSs themselves. For example, if we equip the language from Section \ref{section:example-language} with an operational semantics, do we have full abstraction? Do initial algebras exist for functors written in terms of all or some of the connectives of linear logic? (In other words, do we have inductive types?) What about final co-algebras? In addition, it would be interesting to know how this model relates with probabilistic coherence spaces: do they coincide on countable types? One should also investigate how to get rid of non-effective maps between data types, perhaps by looking for a different (non-free) exponential modality. An other direction would be to extend convex QBSs to a model of probabilistic programming in the broad sense \cite{statonESOP17, conditional19}, that is to say, one capable of describing statistical models rather than just programs that make random choices. Staton's work \cite{statonESOP17} suggests that the first step would be to require stability under integration for all $s$-finite measures (\textit{i.e.} drop the bound on the result) and move the (non-measurable) norm to the structure -- going from convex to ``linear'' QBSs, so to speak. On a different line, replacing $[0,1]$-valued linear tests with tests valued in the unit disc of $\mathbb{C}$ might be a starting point for a model of quantum computation (though this would require leaving the comfort of absolute convergence). Similarly, using tests whose values are intervals included in $[0,1]$, in the spirit of differential program semantics \cite{dlr19}, could yield higher-order versions of such concepts as local differential privacy \cite{locdifpriv1, locdifpriv2}.




\bibliographystyle{IEEEtran}
%

\bibliography{geoffroy-21-convex-qbs}

\end{document}